\documentclass[11pt,a4paper]{article}

\usepackage[utf8]{inputenc}
\usepackage[T1]{fontenc}
\usepackage[margin=2.2cm]{geometry}
\usepackage{amsmath,amssymb,amsthm}
\usepackage{mathtools}
\usepackage{physics}
\usepackage{booktabs}
\usepackage{array}
\usepackage{enumitem}
\usepackage{hyperref}
\usepackage{cleveref}
\usepackage{float}
\usepackage{caption}
\usepackage{xcolor}
\usepackage{fancyhdr}
\usepackage{titlesec}
\usepackage{doi}
\usepackage{tikz}
\usepackage{mathrsfs}
\usetikzlibrary{decorations.pathmorphing, arrows.meta, patterns, calc, positioning}

\hypersetup{colorlinks=true,linkcolor=blue!70!black,citecolor=green!50!black,urlcolor=blue!60!black}

\newcommand{\kap}{\kappa}
\newcommand{\lam}{\lambda}
\newcommand{\rhot}{\tilde{\rho}}

\newtheorem{theorem}{Theorem}[section]
\newtheorem{proposition}[theorem]{Proposition}

\newtheorem{corollary}[theorem]{Corollary}

\theoremstyle{definition}

\newtheorem{remark}[theorem]{Remark}

\title{\LARGE\bfseries Weak-Coupling Limit of the Lattice Nonlinear Schr\"odinger Integral Equation}
\author{Felipe Taha Sant'Ana \\ \small{\emph{Center of Mathematics, Computing, and Cognition, Federal University of ABC}}}
\date{}

\begin{document}
\maketitle
\pagestyle{plain}

\begin{abstract}
We study the ground-state integral equation of the quantum lattice nonlinear
Schr\"odinger model, equivalently the isotropic Heisenberg XXX chain at spin
$s = -1$, whose solution is the root density on a Fermi interval.
In the weak-coupling limit the equation is doubly singular: the driving term and
the kernel collapse onto delta functions, and the solution separates
into three regions on three scales, an inner peak, the bulk of the Fermi sea,
and a boundary layer at the band edge. 
We treat the singular structure probabilistically. We identify both the kernel 
and the driving term as  the Cauchy density, in such a way that the root density corresponds to the 
Green kernel of a Cauchy random walk killed on leaving the Fermi sea. Consequently, the three
regions are then the potential kernel, the interval exit law, and the
ascending-ladder renewal function of that walk.  
This construction determines the logarithmic growth of the central density and the digamma form of the inner profile, whose Fourier-space solution is the Bose--Einstein distribution.  In the bulk, the density approaches the Green function of the half-Laplacian, while the Wiener--Hopf factorisation fixes the edge behaviour and identifies a candidate scale for non-perturbative corrections. The analysis also yields closed asymptotic expressions for the total density and the ground-state energy.  Nystr\"om solutions over a broad range of Fermi rapidities confirm the three-region asymptotic description.
\end{abstract}

\tableofcontents

%==========================================================================
\section{Introduction}\label{sec.intro}
%==========================================================================

The Bethe ansatz, introduced in 1931 for the spin-$\frac{1}{2}$ Heisenberg
chain~\cite{Bethe1931}, provides exact access to the spectrum and thermodynamics of a
wide class of one-dimensional quantum systems.
In the thermodynamic limit, the Bethe equations reduce to linear integral equations
for the density of quasi-momenta, and the extraction of physical observables---particle
density, energy, excitation spectrum---reduces to the analysis of these integral equations
in the relevant parameter regimes~\cite{Gaudin2014,KorepinBogoliubovIzergin1993,Takahashi1999}.

The prototypical example is the Lieb--Liniger model~\cite{LiebLiniger1963,Lieb1963},
describing a one-dimensional Bose gas with repulsive $\delta$-function interactions.
In the thermodynamic limit, its ground state is determined by the integral equation
\begin{equation}\label{eq.LL}
  2\pi\,\rho(\lambda)= 1 + \int_{-q}^{q} K(\lambda -\mu)\,\rho(\mu)\,d\mu\,,
  \qquad K(\lambda) = \frac{2\kappa}{\kappa^2+\lambda^2}\,,
\end{equation}
where $\kappa>0$ is a coupling parameter and $q$ is the Fermi boundary, determined
self-consistently by the particle density.
The Yang--Yang thermodynamic formalism~\cite{YangYang1969} extends it to finite
temperatures.

Analytical developments of iterative solutions of the integral equation were performed in early 
work~\cite{Popov1977} and more recently~\cite{KaminakaWadati2011}, as well as rigorous results for the ground-state
energy~\cite{TracyWidom2016}, while computations of higher-order terms by double
extrapolation have also been achieved~\cite{Prolhac2017}. 
Accurate analytical results and conjectures for the excitation spectrum~\cite{LangHekkingMinguzzi2017}, extensions of the expansion to high order 
and formulation of structural conjectures~\cite{Ristivojevic2014,Ristivojevic2019}, and the structure of both weak- and strong-coupling expansions
for several observables~\cite{Lang2019} have also been accomplished.
These works revealed that the weak-coupling expansion proceeds in half-integer powers
of $\gamma=\kappa/D$, with coefficients involving Riemann zeta values at odd arguments.
Subsequently, exact perturbative
results and demonstration that the series is not Borel summable were obtained~\cite{MarinoReis2019}, establishing a
connection to non-perturbative effects in the dual Gaudin--Yang
model~\cite{Yang1967,Gaudin1967}.

A key mathematical insight~\cite{Gaudin1971} is that the
Lieb--Liniger integral equation belongs to the same family as the Love integral equation
governing the electrostatic potential of two coaxial circular conducting
discs~\cite{Love1949}.
The capacitor problem has a long history in mathematical physics, dating back to
Kirchhoff~\cite{Kirchhoff1877} and Maxwell~\cite{Maxwell1873}, with important
contributions in the early twentieth century~\cite{Nicholson1924,Ignatowsky1932} and thereafter~\cite{Cooke1958,Hutson1963,Sneddon1966}.
The small-separation asymptotics of the capacitor, involving logarithmic terms and
Wiener--Hopf boundary-layer corrections, are directly relevant to the weak-coupling
analysis of the Bethe ansatz.
Recently,~\cite{FarinaLangMartin2022} provided a
comprehensive survey of Love--Lieb integral equations, unifying the electrostatic and
quantum-mechanical perspectives. Moreover, analytical results for the capacitance have been obtained~\cite{ReichertRistivojevic2020}, as well as extension of the Wiener--Hopf analysis to
study non-perturbative corrections for the disc capacitor and the $O(N)$ free energy~\cite{BajnokBalogHegedusVona2022}. 

The discovery that the Lieb--Liniger expansion
is non-Borel summable~\cite{MarinoReis2019} sparked a programme of resurgent analysis in integrable quantum field theories. 
The key tool is the Wiener--Hopf factorisation of the integral equation's symbol
(developed in this context in~\cite{Volin2010,Volin2011}), which gives
systematic access to both perturbative and non-perturbative sectors.
Extending such an analysis to renormalons in integrable field
theories has also been the subject of investigation~\cite{MarinoReis2020a,MarinoReis2020b}, not to mention the Hubbard model~\cite{MarinoReis2022Hubbard}, and the energy gap
problem~\cite{MarinoReis2022gap}.
Important developments were carried out on the $1/N$ expansion~\cite{DiPietroEtAl2021} and detailed
analyses of the $O(4)$ and $O(3)$ sigma models~\cite{AbbottEtAl2021,BajnokBalogVona2022,BajnokBalogVona2022instanton}.
More recently,~\cite{BajnokBalogVona2025} obtained the complete
transseries for conserved charges. 
Such developments are reviewed in~\cite{AnicetoBasarSchiappa2019,Dorigoni2019}.
From this programme the present paper takes only the symbol and the location of
its zeros, which fix the scale of the corrections lying beyond the $1/Q$
expansion.

The quantum inverse scattering method (QISM)~\cite{SklyaninTakhtajanFaddeev1979,Sklyanin1982} provides the algebraic framework
for constructing integrable quantum models on the lattice (see
also~\cite{Faddeev1981,Faddeev1984,Faddeev1995,Faddeev1996} for reviews).
The quantum nonlinear Schr\"odinger (NLS) model on a lattice was constructed
within this framework~\cite{IzerginKorepin1981,IzerginKorepin1982,Kulish1981,Sklyanin1982}, 
with local Hamiltonians worked out in~\cite{TarasovTakhtajanFaddeev1983}.
The lattice NLS provides an integrable discretisation of the continuous NLS field
theory, preserving the full integrability structure at the quantum
level~\cite{KorepinBogoliubovIzergin1993,BogoliubovKorepin1986}.
The model is closely related to the $q$-boson hopping model~\cite{BogoliubovBulloughTimonen1993} 
and it was diagonalised in~\cite{vanDiejen2006}, while~\cite{Dorlas1993} established completeness of the Bethe ansatz eigenstates for the NLS model.

The lattice NLS model is equivalent to the isotropic Heisenberg XXX spin chain
with spin $s=-1$, in the sense of the negative-spin representations recently
introduced~\cite{HaoEtAl2019}. The thermodynamics of the model was also studied recently~\cite{ZhongEtAl2025}.
The higher-spin XXX chain has been long
solved~\cite{KulishReshetikhin1981,KirillovReshetikhin1986}. 
The negative-spin case extends this to non-compact representations, connecting to the
SL(2,$\mathbb{R}$) spin chains that appear in high-energy quantum chromodynamics (QCD)~\cite{KirchManashov2004,DerkachovKorchemskyManashov2001,
DerkachovKorchemskyManashov2002,DerkachovKorchemskyManashov2003}.
Non-compact spin chains have also been studied as integrable stochastic
processes~\cite{FrassekGiardinaKurchan2020},
and the thermodynamic limit of SL(2,$\mathbb{C}$) chains has also been analysed~\cite{GranetJacobsenSaleur2020}.
Building up further its connection to QCD, the $\kappa\to 0$ limit of the XXX$_{s=-1}$ chain corresponds to the
weak-coupling/high-density regime of reggeized gluon dynamics. Furthermore, it was discovered~\cite{Lipatov1993,Lipatov1994,Lipatov1995,Lipatov1997} that the
BFKL (Balitsky-Fadin-Kuraev-Lipatov) pomeron equation---governing small-$x$ parton
evolution~\cite{FadinKuraevLipatov1975,KuraevLipatovFadin1976,KuraevLipatovFadin1977,
BalitskyLipatov1978}---possesses a hidden integrable structure equivalent to a
non-compact Heisenberg spin chain. Ref.~\cite{FaddeevKorchemsky1995} proved the complete integrability of
the BFKL Hamiltonian and identified it with the XXX$_{s=0}$ spin chain, while~\cite{Korchemsky1995,Korchemsky1996,Korchemsky1997} developed the Bethe
ansatz for the QCD pomeron and studied its quasiclassical limit (see
also~\cite{BelitskyBraunGorskyKorchemsky2004,Korchemsky2012} for reviews).
Finally,~\cite{deVegaLipatov2001} extended the analysis to multi-reggeon
interactions.

The ground-state integral equation of the lattice NLS model reads
\begin{equation}\label{eq.latticeNLS}
  2\pi\,\rho(\lambda)
  = K(\lambda)
  + \int_{-q}^{q} K(\lambda -\mu)\,\rho(\mu)\,d\mu\,.
\end{equation}
The crucial difference from~\eqref{eq.LL} is that the driving term is the
Lorentzian $K(\lambda)$---itself degenerating into
$2\pi\delta(\lambda)$ as $\kappa\to 0$---rather than the constant 1.
This makes the weak-coupling limit doubly singular: both the driving term and
the kernel collapse into $\delta$-functions simultaneously, and the techniques
developed for the Lieb--Liniger case---which exploit the regularity of the constant
driving term---do not directly apply.

In this paper we study equation~\eqref{eq.latticeNLS} as $\kappa\to 0$.
We tackle the problem by a rescaling of the spectral variable, a procedure that allows us to  
reduce the problem to an exploitation of three distinct regions. 
The methods draw on classical techniques of singular perturbation
theory~\cite{VanDyke1975,BenderOrszag1999,Hinch1991,Lagerstrom1988} and on the
Wiener--Hopf factorisation technique~\cite{Noble1958,Krein1958,GohbergKrein1958}
as applied to convolution integral equations on finite
intervals~\cite{Sakhnovich2015}. 
Our main findings are as follows.  After rescaling the rapidity as
$\xi=\lambda/\kappa$, the thermodynamic equation depends on the microscopic
parameters only through $Q=q/\kappa$.  Thus, the nontrivial asymptotic regime is
$Q\to\infty$, equivalently the high-filling limit; at fixed spin, the coupling
$\kappa$  sets the rapidity scale. Then, we realise that the rescaled equation admits a
probabilistic representation as both its kernel and its driving term correspond to 
the standard Cauchy density: $\tilde\rho$ is the Green kernel of a Cauchy random
walk killed on leaving $[-Q,Q]$.  The potential kernel of the walk is expressed
through the digamma function, and a stopped-walk identity yields, from a
single construction, the inner profile, the central density, and the outer
limit. 
In addition, we find an identity for the inner profile 
whose full-line Fourier multiplier is the Bose--Einstein distribution
$1/(e^{|p|}-1)$.  On the outer scale $\xi=Qt$, $0<|t|<1$, the same Green
representation yields  the Green function of the Dirichlet half-Laplacian on the
interval with a source at the origin.  Its integral gives the leading density
$D(Q)\sim Q$, while its mean over the Fermi sea is $1/2$.
The subleading density is obtained from the duality with Love's integral
equation and its probabilistic interpretation as the mean exit time of the same
Cauchy walk. Interestingly, the classical small-separation expansion of the circular-disc
capacitor yields the same result. Near
$\xi=\pm Q$, the scaled density is controlled by the ascending-ladder renewal
function $V$ of the Cauchy walk, which, as we will discuss, is connected to the edge limit. 
The corresponding Wiener--Hopf symbol is $1-e^{-|p|}$, and its factorisation
supplies the two-term ladder-renewal expansion; the same expansion, applied at
both barriers, fixes the constant in the total density.
The half-Laplacian limit also determines the low-lying spectrum of the
truncated kernel.
Finally, the correspondence of the driving term with the energy kernel
results in the energy profile \eqref{eq.E_inner_full}.
The complex zeros of the Wiener--Hopf symbol at $p=\pm2\pi i$ further identify
$e^{-2\pi Q}$ as the natural candidate scale for exponentially small
corrections. Nystr\"om solutions of the original
integral equation corroborate the inner, outer, edge, density, and spectral
asymptotics.

The paper is organised as follows.
Section~\ref{sec.model} introduces the model, its Bethe ansatz solution, and the resulting integral equation.
Section~\ref{sec.inner} develops the inner-region analysis, deriving the Bose--Einstein
distribution and the logarithmic divergence.
Section~\ref{sec.outer} treats the outer region and derives the density
expansion; Subsection~\ref{sec.constant} determines the constant $C$ from the
killed Cauchy walk.
Section~\ref{sec.edge} describes the edge boundary layer and the Wiener--Hopf
factorisation that controls it.
Section~\ref{sec.observables} derives the ground-state energy and discusses the physical predictions, including the comparison with the Lieb--Liniger model.
Section~\ref{sec.conclusions} summarises the results and discusses open problems.
The appendices contain further analytical and numerical details: eigenvalue analysis of the truncated kernel
(Appendix~\ref{app.eigenvalues}), the logarithmic moment of the Cauchy exit law
(Appendix~\ref{app.exit-moment}), the next-order outer equation
(Appendix~\ref{app.density}), and the Wiener--Hopf factorisation
(Appendix~\ref{app.wienerhopf}).

%==========================================================================
\section{Lattice nonlinear Schr\"odinger model}
\label{sec.model}
%==========================================================================
We provide the fundamentals of the lattice nonlinear Schr\"odinger model, explain its equivalence to the isotropic Heisenberg
XXX spin chain with spin $s = -1$, derive the Bethe ansatz equations,
and take the thermodynamic limit to obtain the integral
equation~\eqref{eq.latticeNLS}, which is the central object of
the paper.

\subsection{Lattice nonlinear Schr\"odinger model}
\label{sec.lattice-NLS}
The algebraic Bethe ansatz, or quantum inverse scattering method (QISM),
constructs integrable quantum lattice models from
solutions of the Yang--Baxter equation.  The starting point is an
R-matrix $R(\lambda) \in \mathrm{End}(\mathbb{C}^2 \otimes
\mathbb{C}^2)$ satisfying the Yang--Baxter equation
\begin{equation}\label{eq.YBE}
  R_{12}(\lambda - \mu)\,R_{13}(\lambda)\,R_{23}(\mu)
  = R_{23}(\mu)\,R_{13}(\lambda)\,R_{12}(\lambda - \mu)
\end{equation}
in $\mathrm{End}(\mathbb{C}^2 \otimes \mathbb{C}^2 \otimes
\mathbb{C}^2)$.  For the models of interest here, the R-matrix is the
rational (XXX-type) solution
\begin{equation}\label{eq.R-matrix}
  R(\lambda) = \lambda\,\mathbf{I} + i\,\mathbf{P}\,,
\end{equation}
where $\mathbf{I}$ is the identity and $\mathbf{P}$ is the
permutation operator on $\mathbb{C}^2 \otimes \mathbb{C}^2$.
Explicitly, in the standard basis,
\begin{equation}\label{eq.R-explicit}
  R(\lambda) =
  \begin{pmatrix}
    \lambda + i & 0 & 0 & 0 \\
    0 & \lambda & i & 0 \\
    0 & i & \lambda & 0 \\
    0 & 0 & 0 & \lambda + i
  \end{pmatrix}\,.
\end{equation}

The construction of an integrable model on a lattice of $M$ sites
proceeds by choosing, at each site $n$, a local Lax operator
$L_n(\lambda) \in \mathrm{End}(\mathbb{C}^2 \otimes
\mathcal{H}_n)$, where $\mathbb{C}^2$ is the auxiliary space and
$\mathcal{H}_n$ is the local quantum space at site $n$.  The Lax
operator must satisfy the intertwining relation
\begin{equation}\label{eq.RLL}
  R_{12}(\lambda - \mu)\,L_{1n}(\lambda)\,
  L_{2n}(\mu)
  = L_{2n}(\mu)\,L_{1n}(\lambda)\,
  R_{12}(\lambda - \mu)\,,
\end{equation}
which ensures that the monodromy matrix
\begin{equation}\label{eq.monodromy}
  T(\lambda) = L_M(\lambda)\,L_{M-1}(\lambda)
  \cdots L_1(\lambda)
  =:
  \begin{pmatrix}
    A(\lambda) & B(\lambda) \\ C(\lambda) & D(\lambda)
  \end{pmatrix}
\end{equation}
satisfies the RTT relation
\begin{equation}\label{eq.RTT}
  R_{12}(\lambda - \mu)\,T_1(\lambda)\,T_2(\mu)
  = T_2(\mu)\,T_1(\lambda)\,R_{12}(\lambda - \mu)\,.
\end{equation}
The transfer matrix $\tau(\lambda) = \mathrm{tr}_{\mathbb{C}^2}
T(\lambda) = A(\lambda) + D(\lambda)$ then generates a commuting
family: $[\tau(\lambda),\,\tau(\mu)] = 0$ for all $\lambda, \mu$.
The Hamiltonian and higher conserved charges are extracted from the
logarithmic derivative of $\tau(\lambda)$ at a reference
point~\cite{TarasovTakhtajanFaddeev1983, KorepinBogoliubovIzergin1993}.

The local quantum space at each site of the lattice NLS model is a bosonic Fock space
$\mathcal{H}_n = \mathrm{span}\{|0\rangle_n, |1\rangle_n,
|2\rangle_n, \ldots\}$, with creation and annihilation operators
$\psi_n^\dagger$, $\psi_n$ satisfying
$[\psi_n,\,\psi_m^\dagger] = \delta_{nm}$ and number operator
$N_n = \psi_n^\dagger \psi_n$.
The isotropic Heisenberg XXX spin chain with spin
$s$~\cite{KulishReshetikhin1981, KirillovReshetikhin1986} is built from the
local Lax operator
\begin{equation}\label{eq.Lax-spin-s}
  L_n^{(s)}(\lambda) = \lambda\,\mathbf{I}
  + i\sum_{\alpha=1}^{3}\sigma_\alpha \otimes S_n^\alpha
  =
  \begin{pmatrix}
    \lambda + i S_n^z & i S_n^- \\
    i S_n^+ & \lambda - i S_n^z
  \end{pmatrix}\,,
\end{equation}
where $\sigma_\alpha$ are the Pauli matrices (acting on the auxiliary
$\mathbb{C}^2$), $S_n^\alpha$ generate $\mathfrak{sl}(2)$ at site $n$, and
$S^\pm = S^1 \pm iS^2$.  The intertwining relation~\eqref{eq.RLL} with the
rational R-matrix~\eqref{eq.R-matrix} is verified using only the commutation
relations, so it holds for any representation.  For $s = 1/2$ and
$\mathcal{H}_n = \mathbb{C}^2$ this recovers the original Heisenberg
model~\cite{Bethe1931}.
Representation independence is what allows $s$ to be continued to negative
values, provided one works with infinite-dimensional representations of
$\mathfrak{sl}(2)$~\cite{HaoEtAl2019}.  For $s = -1$ the representation is
realised on the bosonic Fock space
$\mathcal{H}_n = \mathrm{span}\{|k\rangle_n : k = 0, 1, 2, \ldots\}$
by the generators
\begin{equation}\label{eq.sl2-bosonic}
  S_n^+ = -(N_n + 2)\,\psi_n\,,\qquad
  S_n^- = \psi_n^\dagger\,,\qquad
  S_n^z = -(N_n + 1)\,.
\end{equation}
A direct computation gives $[S^z, S^\pm] = \pm S^\pm$ and
$[S^+, S^-] = -2(N+1) = 2S^z$, while
$\mathbf{S}^2 = (S^z)^2 + \frac{1}{2}(S^+S^- + S^-S^+) = 0 = s(s+1)$
at $s = -1$.  The spectrum of $S^z$ is $\{-1,-2,-3,\ldots\}$, unbounded below
and bounded above by $-|s|$.  The local Hilbert space is infinite-dimensional,
as it is for every non-compact ($s<0$) representation, whereas the standard
higher-spin chains have finite-dimensional local spaces.  Note that it is
$S^+$, not $S^-$, that annihilates the Fock vacuum, so $|0\rangle_n$ is the
highest-weight vector of the module.

Up to a gauge transformation and a shift of the spectral parameter,
\eqref{eq.Lax-spin-s} with~\eqref{eq.sl2-bosonic} is the L-operator of the
lattice NLS model of Izergin and
Korepin~\cite{IzerginKorepin1981, IzerginKorepin1982, KorepinBogoliubovIzergin1993}. Therefore, in this sense the two models coincide~\cite{HaoEtAl2019, ZhongEtAl2025},
\begin{equation}\label{eq.equivalence}
\text{Lattice NLS model}
  \;\equiv\;
  \text{XXX spin chain with spin $s = -1$}\,.
\end{equation}
In the lattice NLS presentation the model carries an interaction scale
$\kappa$ and a lattice spacing $\Delta$.  These are not independent of the spin:
in the convention of~\cite{HaoEtAl2019} the representation label obeys
\begin{equation}\label{eq.spin-kappa-delta}
  s = -\frac{2}{\kappa\Delta}\,,
\end{equation}
so fixing $s = -1$ imposes $\kappa\Delta = 2$.  Changing
$\kappa$ at fixed $s$ therefore changes $\Delta$ and the rapidity unit together.  
In the spin-chain presentation the same statement reads: the Lax
operator~\eqref{eq.Lax-spin-s} is defined at $\kappa = 1$, and arbitrary
$\kappa$ is restored by the rescaling $\lambda \to \lambda/\kappa$ of the
spectral parameter, which leaves the Yang--Baxter equation invariant and
rescales the width of the kernel in the Bethe equations~\eqref{eq.BAE} by
$\kappa$.  Section~\ref{sec.rescaling} draws the consequence for the
thermodynamic equation.

The local Hamiltonian is obtained from the logarithmic derivative of
$\tau(\lambda)$ at the point where the Lax operator
degenerates~\cite{TarasovTakhtajanFaddeev1983, KorepinBogoliubovIzergin1993},
\begin{equation}\label{eq.H-from-transfer}
  H = c_H\,\frac{d}{d\lambda}\log\tau(\lambda)\Big|_{\lambda = i\kappa}
  + \text{const}\,,
\end{equation}
with the normalisation $c_H$ fixed below by~\eqref{eq.energy-eigenvalue}.
For the bosonic realisation~\eqref{eq.sl2-bosonic} the result is a
nearest-neighbour hopping model with density-dependent
interactions~\cite{TarasovTakhtajanFaddeev1983}: besides the hopping
$-\sum_n(\psi^\dagger_{n+1}\psi_n + \psi^\dagger_n\psi_{n+1})$ and the on-site
repulsion $\frac{\kappa}{2}\sum_n N_n(N_n-1)$, it contains further terms
required by integrability, and it reduces to the continuous quantum NLS field
theory in the continuum
limit~\cite{KorepinBogoliubovIzergin1993, BogoliubovKorepin1986}.  

These are the same non-compact $\mathfrak{sl}(2)$ representations that carry the
integrable structure of high-energy QCD, in the identification reviewed in the
introduction.  Our integral equation~\eqref{eq.latticeNLS} governs the ground state
in the high-density regime, and the logarithmic scaling $\tilde\rho(0;Q) \sim (\log Q)/\pi$ that we find
is qualitatively consistent with the $\log(1/x)$ dependence characteristic of BFKL dynamics.
However, since the BFKL pomeron corresponds to spin $s = 0$ in the
Faddeev--Korchemsky identification~\cite{FaddeevKorchemsky1995} whereas our model has $s = -1$, a precise quantitative
mapping requires a careful continuation in the spin parameter, which is left for future work.

\subsection{Bethe ansatz and thermodynamic limit}
\label{sec.bethe-ansatz}
The algebraic Bethe ansatz diagonalises the transfer matrix
$\tau(\lambda)$ by constructing eigenstates of the form
\begin{equation}\label{eq.Bethe-state}
  |\{\lambda_j\}\rangle = \prod_{j=1}^{N}B(\lambda_j)\,|0\rangle\,,
\end{equation}
where $|0\rangle = |0\rangle_1 \otimes \cdots \otimes |0\rangle_M$ is
the Fock vacuum (all sites empty), $B(\lambda)$ is the off-diagonal
element of monodromy matrix \eqref{eq.monodromy}, and $N$ is the
number of particles.
Since $S^+_n|0\rangle_n = 0$ for the realisation~\eqref{eq.sl2-bosonic}, the
lower-left entry of each $L_n$ in~\eqref{eq.Lax-spin-s} annihilates the
pseudo-vacuum, so $C(\lambda)|0\rangle = 0$, and the diagonal entries act
diagonally with
\begin{equation}\label{eq.vacuum-eigenvalues}
  a(\lambda) = (\lambda - i\kappa)^{M}\,,
  \qquad
  d(\lambda) = (\lambda + i\kappa)^{M}\,,
\end{equation}
where the rescaling $\lambda \to \lambda/\kappa$ has been used to restore
$\kappa$ and $S^z|0\rangle_n = -|0\rangle_n$ has been inserted
in~\eqref{eq.Lax-spin-s}.
The state~\eqref{eq.Bethe-state} is an eigenstate of $\tau(\lambda)$ provided
$a(\lambda_j)/d(\lambda_j) = \prod_{k\neq j}
(\lambda_j-\lambda_k+i\kappa)/(\lambda_j-\lambda_k-i\kappa)$, that is, provided
the rapidities $\{\lambda_j\}_{j=1}^N$ satisfy the Bethe ansatz equations
(BAE)~\cite{KorepinBogoliubovIzergin1993, HaoEtAl2019, ZhongEtAl2025}
\begin{equation}\label{eq.BAE}
  \left(\frac{\lambda_j + i\kappa}{\lambda_j - i\kappa}\right)^M
  = \prod_{\substack{k=1 \\ k \neq j}}^{N}
  \frac{\lambda_j - \lambda_k - i\kappa}
       {\lambda_j - \lambda_k + i\kappa}\,,
  \qquad j = 1, \ldots, N\,.
\end{equation}
The half-width of the vacuum factor on the left is $\kappa$, the same as that
of the scattering factor on the right.  For general spin the two differ---the
vacuum factor has half-width $|s|\kappa$---and the correspondence at $s = -1$ is
the origin of every structural peculiarity discussed in this paper.

The energy eigenvalue follows from~\eqref{eq.H-from-transfer}.  As for all
XXX-type chains, the bare energy of a rapidity is minus the derivative of its
bare momentum~\cite{TarasovTakhtajanFaddeev1983, KorepinBogoliubovIzergin1993};
fixing the normalisation $c_H$ so that the proportionality constant is unity,
\begin{equation}\label{eq.energy-eigenvalue}
  E = -\sum_{j=1}^{N}\frac{2\kappa}{\lambda_j^2 + \kappa^2}
  + \text{const}
  = -\sum_{j=1}^{N} K(\lambda_j;\,\kappa) + \text{const}\,,
\end{equation}
with $K$ the Lorentzian defined in~\eqref{eq.LL}.  This same
function controls the bare momentum, the two-body scattering phase
and the bare energy.

We now take the thermodynamic limit $M, N \to \infty$ with fixed
density $D = N/M$ and fixed coupling $\kappa$.  In this limit, the
Bethe roots $\{\lambda_j\}$ become dense on an interval $[-q, q]$
(the Fermi sea), with a distribution function $\rho(\lambda)$
normalised so that
\begin{equation}\label{eq.density-normalisation}
  \int_{-q}^{q}\rho(\lambda)\,d\lambda = D = \frac{N}{M}\,.
\end{equation}
The Fermi boundary $q$ is determined self-consistently by the particle
density.
The passage to the continuum proceeds as follows.  Taking the logarithm
of~\eqref{eq.BAE} and introducing the bare momentum and the two-body phase
\begin{equation}\label{eq.bare-phase}
  p_0(\lambda) = i\log\frac{\lambda + i\kappa}{\lambda - i\kappa}\,,
  \qquad
  \theta(\lambda) = p_0(\lambda)\,,
  \qquad
  p_0'(\lambda) = \theta'(\lambda) = \frac{2\kappa}{\kappa^2+\lambda^2}\,,
\end{equation}
(the two functions coincide for $s=-1$, by the remark following~\eqref{eq.BAE}),
the Bethe equations read
\begin{equation}\label{eq.BAE-log}
  M\,p_0(\lambda_j) + \sum_{k \neq j}^{N}\theta(\lambda_j - \lambda_k)
  = 2\pi I_j\,,
\end{equation}
with distinct quantum numbers $I_j$, consecutive in the ground state.  The
counting function
\begin{equation}\label{eq.counting-function}
  z(\lambda) = \frac{1}{2\pi}\left[p_0(\lambda)
  + \frac{1}{M}\sum_{k=1}^{N}\theta(\lambda - \lambda_k)\right]\,,
  \qquad z(\lambda_j) = \frac{I_j}{M}\,,
\end{equation}
is monotonic and interpolates the quantum numbers.  As $M, N \to \infty$ the
roots become dense and the sums converge to integrals against $\rho$,
$M^{-1}\sum_k F(\lambda_k) \to \int_{-q}^{q}F(\mu)\rho(\mu)\,d\mu$ for
continuous $F$; the spacing of consecutive $I_j$ is unity, so
$z'(\lambda) = \rho(\lambda) + \rho_h(\lambda)$, and in the ground state every
quantum number inside the Fermi sea is occupied, $\rho_h \equiv 0$ on $[-q,q]$.
Differentiating~\eqref{eq.counting-function} and
using~\eqref{eq.bare-phase} gives the ground-state integral equation
\eqref{eq.latticeNLS}.  Its Lorentzian kernel, defined in~\eqref{eq.LL}, is
normalised so that
\begin{equation}\label{eq.K_normalisation}
  \int_{-\infty}^{\infty}K(\lambda;\,\kappa)\,d\lambda = 2\pi
\end{equation}
for all $\kappa > 0$.
This is a Fredholm integral equation of the second kind on the symmetric interval $[-q, q]$, with
a Lorentzian kernel that is symmetric, positive, and of convolution type.

The crucial structural difference from the Lieb--Liniger integral equation~\eqref{eq.LL}
lies in the driving term: in the Lieb--Liniger equation, the driving
term is the constant $1$, arising from the unrestricted quadratic dispersion of the continuous model; in the lattice NLS equation~\eqref{eq.latticeNLS}, the
driving term is the Lorentzian $K(\lambda;\,\kappa)$, reflecting the bounded bandwidth of the lattice momentum.
In the limit $\kappa \to 0$, $K(\lambda;\,\kappa) \to 2\pi\,\delta(\lambda)$ in the distributional sense, so the driving
term degenerates into a $\delta$-function.  This is the source of the
doubly singular structure that makes the weak-coupling analysis
of~\eqref{eq.latticeNLS} qualitatively different from that
of~\eqref{eq.LL}. Such a singular behaviour necessitates the methods developed in the following sections\footnote{
The kernel $K(\lambda - \mu;\,\kappa)$ and the driving term
$K(\lambda;\,\kappa)$ are both even functions, so the
solution $\rho(\lambda)$ is even: $\rho(-\lambda) = \rho(\lambda)$.
Since $K \geq 0$ and the driving term $K > 0$, the Neumann
series converges and $\rho(\lambda) > 0$ for all
$\lambda \in [-q,q]$~\cite{KorepinBogoliubovIzergin1993}.
The Fourier transform of the Lorentzian kernel is
\begin{equation}\label{eq.K-fourier}
  \hat K(p;\,\kappa) = \int_{-\infty}^{\infty}e^{-ip\lambda}\,
  \frac{2\kappa}{\kappa^2 + \lambda^2}\,d\lambda
  = 2\pi\,e^{-\kappa|p|}\,.
\end{equation}
This exponential decay in momentum space is crucial: the kernel is
almost critical in the sense that $\hat K(0;\,\kappa) = 2\pi$
equals the prefactor on the left-hand side
of~\eqref{eq.latticeNLS}, making the homogeneous equation
$2\pi\rho = K * \rho$ have eigenvalue~$1$ on the full line.  On the
finite interval $[-q,q]$, the largest eigenvalue is strictly less
than $2\pi$, but approaches it as $q/\kappa \to \infty$---which is
precisely the weak-coupling limit.}.
 
Once the density $\rho(\lambda)$ is known, the physical observables of the
ground state follow.  The number of particles per site
is~\eqref{eq.density-normalisation}, and the energy per site in the
thermodynamic limit follows from~\eqref{eq.energy-eigenvalue}:
\begin{equation}\label{eq.energy-per-site}
  e(\kap) = -\int_{-q}^{q}\frac{2\kappa}{\kappa^2 + \lambda^2}\,
  \rho(\lambda)\,d\lambda
  = -\int_{-q}^{q}K(\lambda;\,\kappa)\,\rho(\lambda)\,d\lambda\,.
\end{equation}
The kernel appearing in the energy integral is the same function
$K(\lambda;\,\kappa)$ that serves as both the driving term and the
scattering kernel in~\eqref{eq.latticeNLS}---a correspondence
specific to the lattice NLS model that will have important
consequences\footnote{See the identity
\eqref{eq.E_inner} in
Section~\ref{sec.observables}.}.
The total momentum per site is $P = 0$ for the ground state, by parity.
Excited states are described by particle--hole excitations above the
Fermi sea, with dressed energy and dressed momentum satisfying linear integral equations of the same
type as~\eqref{eq.latticeNLS}, with different driving
terms~\cite{KorepinBogoliubovIzergin1993, YangYang1969}.  The excitation spectrum and thermodynamics of
the XXX$_{s=-1}$ chain were recently studied~\cite{ZhongEtAl2025}. In this manuscript, we restrict ourselves to the ground-state.

%==========================================================================
\section{Inner region}
\label{sec.inner}
%==========================================================================
\subsection{Rescaled inner equation}
\label{sec.rescaling}

To resolve the singularity, we introduce the rescaled variables
\begin{equation}\label{eq.rescaling}
  \xi = \frac{\lambda}{\kappa}\,,\qquad
  \tilde\rho(\xi) = \kappa\,\rho(\kappa\xi)\,,\qquad
  Q = \frac{q}{\kappa}\,,
\end{equation}
which zoom in on the region $|\lambda| \sim \kappa$ where the
driving term is concentrated.  Since
$\tilde\rho(\xi)\,d\xi = \kappa\rho(\kappa\xi)\,d\lambda/\kappa = \rho(\lambda)\,d\lambda$,
the rescaling preserves the particle density,
\begin{equation}\label{eq.density-rescaled}
  \int_{-Q}^{Q}\tilde\rho(\xi)\,d\xi = \int_{-q}^{q}\rho(\lambda)\,d\lambda = D\,.
\end{equation}
Under this change of variables,
equation~\eqref{eq.latticeNLS} becomes
\begin{equation}\label{eq.inner-eq}
2\pi\,\tilde\rho(\xi) = \frac{2}{1 + \xi^2}
  + \int_{-Q}^{Q}\frac{2}{1 + (\xi - \eta)^2}\,
  \tilde\rho(\eta)\,d\eta\,.
\end{equation}
This is a Lieb--Liniger-type equation on the
expanding domain $[-Q,Q]$, but with the constant driving term
 replaced by the Lorentzian 
\begin{equation}\label{eq.rescaled_kernel}
    K(\xi) = \frac{2}{1+\xi^2}
\end{equation}
and the coupling
parameter set to $\kappa = 1$.  The limit $\kappa \to 0$ at fixed $q$
corresponds to $Q = q/\kappa \to \infty$.
Together with the density relation~\eqref{eq.density-rescaled}, the energy per
site transforms as
\begin{equation}\label{eq.observables-rescaled}
  e(\kap) = -\frac{1}{\kappa}\int_{-Q}^{Q}\frac{2}{1+\xi^2}\,
  \tilde\rho(\xi)\,d\xi\,.
\end{equation}
Equation~\eqref{eq.inner-eq} depends on $\kappa$ and $q$ only through $Q$.  Two
consequences are worth stating at the outset.
First, the particle density $D = D(Q)$ is a function of $Q$ alone, so holding
$D$ fixed holds $Q$ fixed; the limit $\kappa\to0$ at fixed density therefore
does not act on the integral equation at all, and the residual $1/\kappa$
in~\eqref{eq.observables-rescaled} is the overall normalisation of $H$ inherited
from~\eqref{eq.H-from-transfer}.
Second, the limit that does act on~\eqref{eq.inner-eq} is $Q\to\infty$, reached
either by $\kappa\to0$ at fixed Fermi rapidity $q$ or by increasing the filling
at fixed $\kappa$; in both cases $D(Q)\to\infty$.
Every expansion below is an expansion in $1/Q$, and statements phrased as
$\kappa\to0$ are to be read at fixed $q$.
The analysis of~\eqref{eq.inner-eq} as $Q \to \infty$ is the subject
of the remainder of this paper.

In the limit $Q\to\infty$ the integration
domain can be extended to the whole real line, and the inner equation~\eqref{eq.inner-eq}
takes the convolution form
\begin{equation}
2\pi\,\tilde{\rho}(\xi)
=
g(\xi)+(K*\tilde{\rho})(\xi),
\end{equation}
where
\begin{equation}
    g(\xi)=K(\xi), \qquad    (K*\tilde{\rho})(\xi) = \int_{-\infty}^{\infty}
K(\xi-\eta)\tilde{\rho}(\eta)\,d\eta\, .
\end{equation}

\subsection{Fourier analysis and logarithmic divergence}
We use the Fourier transform convention
\begin{equation}
\widehat{f}(p)
=
\int_{-\infty}^{\infty} e^{-ip\xi}f(\xi)\,d\xi,
\qquad
f(\xi)
=
\frac{1}{2\pi}
\int_{-\infty}^{\infty}e^{ip\xi}\widehat{f}(p)\,dp\,,
\end{equation}
so that the Fourier transform of a convolution is the product of Fourier transforms.
Equation~\eqref{eq.K-fourier} at $\kappa=1$, together with $g=K$, gives
$\widehat{K}(p)=\widehat{g}(p)=2\pi e^{-|p|}$.
Applying the Fourier transform to the full-line equation yields
\begin{equation}
2\pi\,\hat{\tilde{\rho}}(p)
=
2\pi e^{-|p|}
+
2\pi e^{-|p|}
\hat{\tilde{\rho}}(p),
\end{equation}
where the convolution theorem has been used.
Solving algebraically:
\begin{equation}\label{eq.BE-distribution}
\hat{\tilde{\rho}}(p)
=
\frac{e^{-|p|}}{1-e^{-|p|}}
=
\frac{1}{e^{|p|}-1}\,.
\end{equation}
This is the Bose--Einstein distribution.
It presents a $1/|p|$ singularity at $p=0$:
\begin{equation}\label{eq.BE-smallp}
\frac{1}{e^{|p|}-1}
=
\frac{1}{|p|}
-\frac{1}{2}
+\frac{|p|}{12}
+\mathcal O\left(|p|^3\right)\,.
\end{equation}
Consequently, the inverse Fourier transform
\begin{equation}\label{eq.rho_xi_cos}
\rhot^{(\infty)}(\xi) = \frac{1}{\pi}\int_0^{\infty}\frac{\cos(\xi p)}{e^p-1}\,dp
\end{equation}
diverges logarithmically at $\xi=0$.
This divergence is regularised by the finite domain $[-Q,Q]$, producing the logarithmic dependence on~$Q$.

To isolate the divergence we introduce a lower cutoff $\varepsilon>0$ and
consider 
\begin{equation}
    I(\varepsilon) = \int_{\varepsilon}^{\infty}dp  \, \frac{1}{e^{p}-1}\, .
\end{equation}
Using the geometric series
\begin{equation}\label{eq.geometric-series}
    \frac{1}{e^{p}-1}=\sum_{n=1}^{\infty}e^{-np}
\end{equation}
and integrating term by term, we obtain
\begin{equation}
    I(\varepsilon) = \sum_{n=1}^{\infty}e^{-n\varepsilon}/n = -\log(1-e^{-\varepsilon})\, .
\end{equation}
For small $\varepsilon$,
\begin{equation}\label{eq.Bose-log}
  I(\varepsilon)=-\log\varepsilon + 
  \mathcal O (1)\,,
  \qquad \varepsilon\to 0\,.
\end{equation}

\begin{figure}[htbp]
\centering
\includegraphics[width=\textwidth]{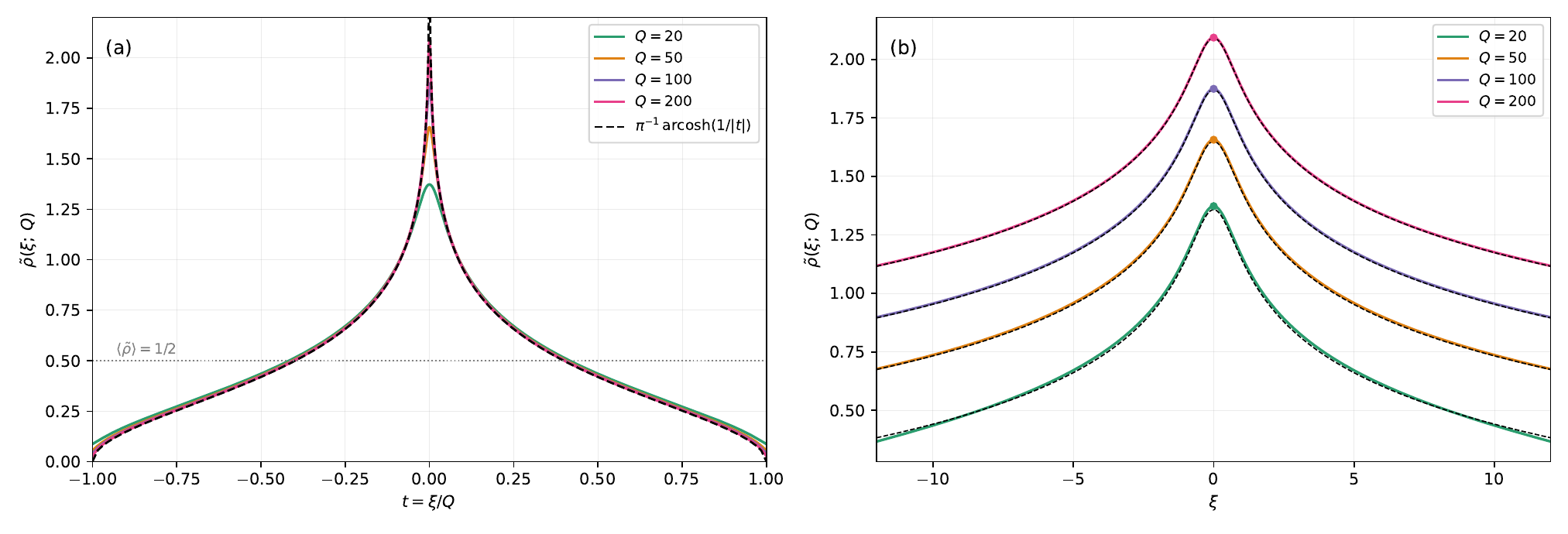}
\caption{Numerical solution of the rescaled integral equation~\eqref{eq.inner-eq}
for $Q = 20$, $50$, $100$, and $200$.
\textbf{(a)}~Rescaled density $\tilde\rho(\xi;\,Q)$ against $t = \xi/Q$, showing
the full domain.  Outside the inner region the curves collapse onto the outer
profile~\eqref{eq.outer-solution} (dashed), which is nowhere constant: it falls monotonically from the inner peak
and vanishes as $\sqrt{1-|t|}$ at the band edge.  Its mean over the Fermi sea is
$1/2$ (dotted line), and it crosses that value at $|t| = 1/\cosh(\pi/2)$.
The residual deviation from the dashed curve is $\mathcal O(\log Q/Q)$ and
accounts for the subleading terms in~\eqref{eq.D-expansion}.
\textbf{(b)}~Zoom into the inner region $|\xi| \lesssim \mathcal O(1)$, plotted
against the rescaled variable~$\xi$.  The peak height at $\xi = 0$
increases as $\tilde\rho(0;\,Q) \sim (\log Q)/\pi + C$.
The cusp-like shape reflects the $1/|p|$ infrared singularity of the
Bose--Einstein distribution~\eqref{eq.BE-distribution} in Fourier space, and the
profile is controlled by the digamma function $\operatorname{Re}\psi(1+i\xi)$:
the dashed curves show the inner approximation
$\tilde\rho(\xi;Q) \approx [\log(2Q) - \operatorname{Re}\psi(1+i\xi)]/\pi$
from~\eqref{eq.inner-profile-relative}, which matches the numerical solution over the range shown.}
\label{fig.solution-profile}
\end{figure}

\subsection{Eigenvalue analysis of the truncated kernel}
\label{sec.eigenvalue-summary}
The integral operator $\mathcal{K}_Q$ on $L^2([-Q,Q])$ defined by
\begin{equation}\label{eq.KQ-def}
  (\mathcal{K}_Q f)(\xi) = \int_{-Q}^{Q}\frac{2}{1+(\xi-\eta)^2}\,f(\eta)\,d\eta
\end{equation}
has a discrete spectrum $\{\lam_n(Q)\}_{n=0}^{\infty}$ with $\lam_0(Q)>\lam_1(Q)>\cdots>0$.

On the full line, $\hat{\mathcal{K}}(p)=2\pi e^{-|p|}$, meaning that the operator norm is $2\pi$. For finite $Q$, the largest eigenvalue $\lam_0(Q)$ approaches $2\pi$ from below:
$\lam_0(Q) = 2\pi - \Delta_0(Q)$, with $\Delta_0(Q)\to 0$ as $Q\to\infty$.
The rate and the constant both follow from the outer limit of
Section~\ref{sec.outer}.  Since the normalized  symbol $\sigma(p):=1-e^{-|p|}$ satisfies $\sigma(p) \to |p|$ as $p \to 0$ and the
low-lying eigenfunctions live at $|p| = \mathcal O(1/Q)$, the eigenvalue problem
for $2\pi I - \mathcal{K}_Q$ reduces to that of the half-Laplacian
$\mathcal{A}$ on $(-Q,Q)$, whose eigenvalues are $\lambda_n/Q$ by homogeneity,
with $\lambda_n$ those of $\mathcal{A}$ on $(-1,1)$.  Hence
\begin{equation}\label{eq.gap-identified}
  \Delta_n(Q) = \frac{2\pi\lambda_{n+1}}{Q} + o\left(Q^{-1}\right)\,.
\end{equation}
The $\lambda_n$ are known~\cite{Kwasnicki2012}; Appendix~\ref{app.eigenvalues}
collects them, together with the computed spectrum of $\mathcal K_Q$ against
which~\eqref{eq.gap-identified} is checked.

The closing of the lowest spectral gap should not, however, be identified by
itself with the logarithmic growth of the central density.  Indeed,
\eqref{eq.gap-identified} shows that the first eigenvalue is not asymptotically
isolated: for every fixed $n$ the corresponding deficit is of order $Q^{-1}$,
and in particular $\frac{\Delta_1(Q)}{\Delta_0(Q)} \longrightarrow \frac{\lambda_2}{\lambda_1}.$
Thus an increasing family of low-lying modes approaches the critical value
$2\pi$ as $Q\to\infty$.  This many-mode accumulation is consistent with the
logarithmic growth of $\tilde\rho(0;Q)$, which is derived from the
killed-Cauchy-walk representation in Section~\ref{sec.constant}; the deficits
themselves follow from the form convergence proved in
Appendix~\ref{app.eigenvalues}.  Further details on the spectrum
of $\mathcal K_Q$ are given in Appendix~\ref{app.eigenvalues}, and the closing
of the first two deficits is illustrated in Figure~\ref{fig.spectral-gap}.

\begin{figure}[htbp]
\centering
\includegraphics[width=\textwidth]{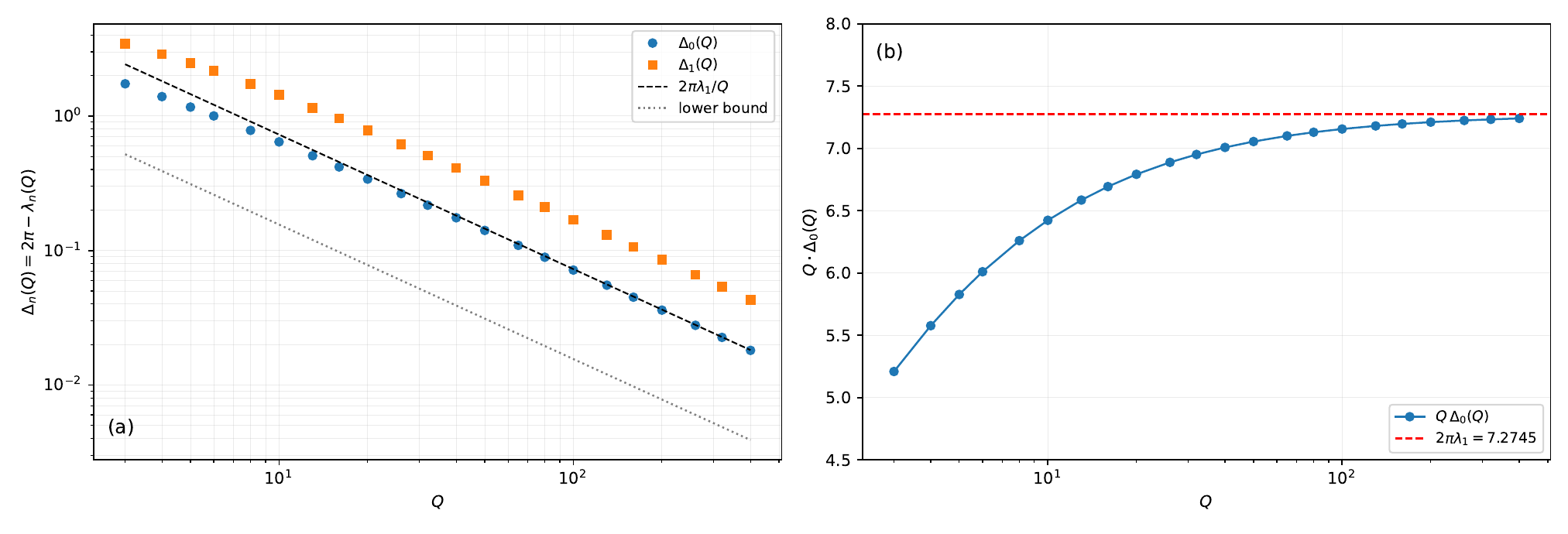}
\caption{Spectral gap of the truncated kernel $\mathcal{K}_Q$.
\textbf{(a)}~Log--log plot of $\Delta_n(Q) = 2\pi - \lambda_n(Q)$
for $n=0$ (circles) and $n=1$ (squares) versus~$Q$.  Both gaps close
algebraically, with $\Delta_0$ below $\Delta_1$ for all~$Q$.  The dotted line is
the lower bound of~\eqref{eq.gap-upper}; the dashed line is the sharp upper
bound $2\pi\lambda_1/Q$, with $\lambda_1$ the first Dirichlet eigenvalue of the
half-Laplacian on $(-1,1)$.
\textbf{(b)}~Compensated gap $Q\cdot\Delta_0(Q)$ versus~$Q$ on a
semi-logarithmic scale.  The product approaches
$2\pi\lambda_1 = 7.2745$ (dashed) from below, so the gap closes as
$2\pi\lambda_1/Q$ with no logarithmic modulation; the apparent slow growth over
a restricted range of $Q$ is a pre-asymptotic effect.}
\label{fig.spectral-gap}
\end{figure}

%==========================================================================
\section{Outer region and density decomposition}
\label{sec.outer}
%==========================================================================

\subsection{Outer solution}
To resolve the density on the scale of the whole Fermi sea, we introduce the
outer variable
\begin{equation}
    t=\frac{\xi}{Q},\qquad -1<t<1,
\end{equation}
and consider $Q\to\infty$ with $t\neq0$ fixed.  The overlap with the inner
region is recovered subsequently by taking $|t|\ll 1$, or equivalently
$1\ll|\xi|\ll Q$.
The limiting operator follows directly from the low-momentum behaviour of the
convolution kernel.  On the full line, the operator $2\pi I-\mathcal K$ has
Fourier symbol
\begin{equation}\label{eq:outer-symbol}
2\pi\left(1-e^{-|p|}\right)
=2\pi|p|-\pi p^2+\mathcal O(|p|^3),
\qquad p\to0.
\end{equation}
On the outer scale the momentum conjugate to $t$ is $k=Qp$, and hence
$  2\pi\left(1-e^{-|k|/Q}\right)
  =\frac{2\pi}{Q}|k|+\mathcal O \left(Q^{-2}k^2\right)$.
Thus, the leading operator is the  half-Laplacian $\mathcal A=(-d^2/dt^2)^{1/2}$,
whose Fourier symbol is $|k|$ .
The driving term has a nontrivial limit on the same scale.  Although
$g(Qt)=2/(1+Q^2t^2)$ tends pointwise to zero for $t\neq0$, its mass remains
concentrated at the origin:
\begin{equation}
      Q\,g(Qt)
  =\frac{2Q}{1+Q^2t^2}
  \;\longrightarrow\;
  2\pi\,\delta(t) \, .
\end{equation}
Consequently the leading outer equation is
\begin{equation}\label{eq:outer-problem-scaled}
\mathcal A \, \tilde\rho_{\rm out}(t)=\delta(t),
\qquad -1<t<1,
\end{equation}
with the exterior Dirichlet condition associated with the finite Fermi
interval.  The outer density is therefore the Green function of the
Dirichlet half-Laplacian on $(-1,1)$ with a source at the origin.

Writing $\mathcal{A}$ for the operator with symbol $|p|$, equivalently
\begin{equation}\label{eq.A-operator}
  (\mathcal{A}f)(\xi)
  = -\frac{1}{\pi}\,\mathrm{p.v.}\int \frac{f'(\eta)}{\xi-\eta}\,d\eta\,,
\end{equation}
the solution of~\eqref{eq:outer-problem-scaled}, written in the original
rescaled coordinate $\xi=Qt$, is
\begin{equation}\label{eq.outer-solution}
  \tilde\rho_{\mathrm{out}}(\xi)
  = \frac{1}{\pi}\operatorname{arcosh}\frac{Q}{|\xi|}
  = \frac{1}{\pi}\log\frac{Q + \sqrt{Q^2-\xi^2}}{|\xi|}\,.
\end{equation}

Note that the symbol $|p|$ is the half-Laplacian
$\mathcal{A} = (-d^2/d\xi^2)^{1/2}$, taken with the exterior Dirichlet condition
$\tilde\rho = 0$ outside $[-Q,Q]$. Equivalently, it is the generator of the
Cauchy process killed upon leaving the Fermi interval~\cite{BanuelosKulczycki2004, KulczyckiKwasnickiMaleckiStos2010}.
Its Green function on $(-Q,Q)$ is 
\begin{equation}\label{eq.green}
  G(\xi,\eta) = \frac{1}{\pi}\,
  \log\left|\frac{Q^2 - \xi\eta
    + \sqrt{(Q^2-\xi^2)(Q^2-\eta^2)}}{Q\,(\xi-\eta)}\right|\,,
\end{equation}
and~\eqref{eq.outer-solution} is nothing but $G(\xi,0)$: the outer solution is
the Green function of the half-Laplacian with the source at the origin, which is
where the driving term deposits its weight.   

Here it convenient to write explicitly two properties of~\eqref{eq.outer-solution} that will be used in the following.  Its
integral is the leading density,
\begin{equation}\label{eq.outer-integral}
  \int_{-Q}^{Q}\tilde\rho_{\mathrm{out}}(\xi)\,d\xi
  = \frac{2Q}{\pi}\int_0^{1}\operatorname{arcosh}\frac{1}{t}\,dt
   = Q\,,
\end{equation}
so that the mean of the profile across the Fermi sea is
\begin{equation}\label{eq.bulk-value}
  \bar{\tilde\rho}_{\mathrm{bulk}} = \frac{1}{2}\,.
\end{equation}

We next compare the outer solution with the inner behaviour obtained in
Section~\ref{sec.inner}.  The full-line Fourier analysis determines the inner
profile relative to its value at the origin,
\begin{equation}\label{eq.inner-profile-relative}
\tilde\rho(\xi;Q)-\tilde\rho(0;Q)
=
-\frac{\operatorname{Re}\psi(1+i\xi)+\gamma_E}{\pi}
+o(1),
\end{equation}
for fixed $\xi$ as $Q\to\infty$.  Since
$\operatorname{Re}\psi(1+i\xi)=\log|\xi|+\mathcal O(\xi^{-2})$, the inner
solution develops the logarithmic dependence $-\pi^{-1}\log|\xi|$ in the overlap
region $1\ll|\xi|\ll Q$.  The outer solution~\eqref{eq.outer-solution} has, for
$|\xi|\ll Q$,
\begin{equation}\label{eq:outer-local}
\tilde\rho_{\mathrm{out}}(\xi)
=
\frac{1}{\pi}\log\frac{2Q}{|\xi|}
+\mathcal O\left(\frac{\xi^2}{Q^2}\right).
\end{equation}
The two descriptions therefore carry the same logarithmic
dependence on $|\xi|$ in their common range of validity. 
The finite part of the correspondence is tied to the constant term in $\tilde\rho(0;Q)$, which we derive  directly from the finite-interval problem in Section~\ref{sec.constant}.

Figure~\ref{fig.solution-profile}(a) shows the convergence of the numerical
solutions to the outer profile~\eqref{eq.outer-solution} in the variable
$t=\xi/Q$.  The remaining finite-$Q$ correction is spread across the Fermi sea
and supplies the subleading terms in the total-density expansion discussed
below.

\subsection{Duality with the Love integral equation}
\label{sec.duality-Love}

The integral equation~\eqref{eq.inner-eq} is closely related to the
Love integral equation~\cite{Love1949} for the electrostatic potential
of two coaxial circular conducting discs.  Rescaling $\xi = Qx$ maps
$[-Q,Q]$ onto $[-1,1]$ and turns the kernel below into
$(\kappa_L/\pi)\left[\kappa_L^2+(x-y)^2\right]^{-1}$ with $\kappa_L = 1/Q$, so
$Q^{-1}$ is the plate separation in units of the disc radius.  In the variable
$\xi$ the Love equation reads
\begin{equation}\label{eq.Love}
  f(\xi) = 1 + \frac{1}{\pi}\int_{-Q}^{Q}
  \frac{f(\eta)}{1 + (\xi - \eta)^2}\,d\eta\,.
\end{equation}
We define the Love operator $\mathcal{L}_Q$ by
\begin{equation}\label{eq.LQ-def}
    (\mathcal{L}_Q h)(\xi) = \frac{1}{\pi}\int_{-Q}^Q
\frac{h(\eta)}{1+(\xi-\eta)^2}\,d\eta \, ,
\end{equation}
so that our rescaled equation~\eqref{eq.inner-eq} takes the form
\begin{equation}\label{eq.lattice-love-form}
    (I - \mathcal{L}_Q)\tilde\rho = \frac{g}{2\pi}, \qquad 
g(\xi) = 2/(1+\xi^2)\, ,
\end{equation}
while the Love equation is $(I - \mathcal{L}_Q)f = 1$.
The two equations share the same operator $I - \mathcal{L}_Q$ but
have different right-hand sides: the Lorentzian $g/(2\pi)$
in the lattice NLS case, versus the constant 1 
in the Love case.

Since $\mathcal{L}_Q$ is self-adjoint on $L^2([-Q,Q])$ (the kernel
$K(\xi-\eta)/(2\pi)$ is real and symmetric in $\xi,\eta$), the
resolvent $(I - \mathcal{L}_Q)^{-1}$ is also self-adjoint.
Writing $\tilde\rho = (I - \mathcal{L}_Q)^{-1}[g/(2\pi)]$ and
$f = (I - \mathcal{L}_Q)^{-1}\mathbf{1}$, the total density becomes
\begin{equation}\label{eq.duality}
  D(Q)
  = \langle \mathbf{1}, \tilde\rho\rangle
  = \left\langle (I - \mathcal{L}_Q)^{-1}\mathbf{1},\,
      g/(2\pi)\right\rangle
  = \left\langle f,\,\frac{g}{2\pi}\right\rangle
  = \frac{1}{\pi}\int_{-Q}^{Q} \frac{f(\xi)}{1+\xi^2}\,d\xi\,,
\end{equation}
where the third equality uses the self-adjointness of the resolvent.
This identity reduces the lattice NLS density to a weighted integral of
the Love solution---the same function that determines the capacitance of
the circular disc capacitor~\cite{Love1949,FarinaLangMartin2022}.

A further identity follows from evaluating the Love equation
at $\xi = 0$.  Since $K(0-\eta) = K(\eta) = 2/(1+\eta^2)$, the
convolution integral at the origin becomes
\begin{equation}\label{eq.love-peak}
  f(0)
  = 1 + \frac{1}{\pi}\int_{-Q}^{Q}\frac{f(\eta)}{1+\eta^2}\,d\eta
  = 1 + D(Q)\,,
\end{equation}
where the last step uses~\eqref{eq.duality}. The Love function has a
probabilistic reading of its own, which we use in
Section~\ref{sec.edge}.  Expanding
$f = (I-\mathcal{L}_Q)^{-1}1$ in a Neumann series and recognising
$(\mathcal{L}_Q^n 1)(\xi) = \mathbb{P}_\xi(\tau_Q > n)$, the probability that
the walk of Section~\ref{sec.matching-C} survives $n$ steps inside the Fermi
sea,
\begin{equation}\label{eq.love-exit-time}
  f(\xi) = \sum_{n\geq0}\mathbb{P}_\xi(\tau_Q > n) = \mathbb{E}_\xi\,\tau_Q\,,
  \qquad\text{so}\qquad
  D(Q) = \mathbb{E}_0\,\tau_Q - 1\,.
\end{equation}
The total density is the mean number of steps the Cauchy walk takes to escape
between two absorbing barriers, and the constant $b$ is the constant term of
that mean duration---a two-barrier fluctuation problem, first connected to the
capacitor by Reich~\cite{Reich1953}.  This is an exact,
non-asymptotic identity valid for all $Q > 0$, linking the peak of the
Love solution directly to the lattice NLS density.

The parameter $Q$ is the inverse plate separation, so the limit
$Q \to \infty$ studied here is the limit of closely spaced plates.  That
is the regime in which the capacitance develops the logarithmic corrections
first computed by Kirchhoff~\cite{Kirchhoff1877} and
Maxwell~\cite{Maxwell1873}, with Wiener--Hopf corrections analysed
in~\cite{Hutson1963}.

\subsection{Determination of \texorpdfstring{$C$}{C}}
\label{sec.constant}
%==========================================================================
We now proceed to determining  the constant in the asymptotic expansion
\begin{equation}\label{eq.rho0-expansion}
  \tilde\rho(0;\,Q) = \frac{\log Q}{\pi} + C + o(1)\,.
\end{equation}
The derivation is probabilistic.  The kernel of~\eqref{eq.inner-eq} is a
probability density, so $\tilde\rho$ is a Green kernel for a random walk killed
on leaving $[-Q,Q]$, and the constant is the logarithmic moment of the exit
position.

\subsubsection{Killed Cauchy walk}
\label{sec.matching-C}

Set $P(\xi):=g(\xi)/(2\pi)=1/[\pi(1+\xi^2)]$.  Then
\eqref{eq.lattice-love-form}, with the operator~\eqref{eq.LQ-def}, is a
resolvent equation driven by $P$.
The observation that organises everything below is that $P$ is a probability
density---the standard Cauchy density---and that the driving term is that same
density.  Let $\xi_1,\xi_2,\ldots$ be independent with density $P$, let
$S_n = \xi + \xi_1 + \cdots + \xi_n$ be the walk started at $\xi$, and let
\begin{equation}\label{eq.exit-time}
  \tau_Q = \inf\{n \geq 1 : |S_n| > Q\}
\end{equation}
be the first exit time from the Fermi sea.  Then $(I-\mathcal{L}_Q)^{-1}P$ is
the sum over $n$ of the transition densities of the walk killed at $\tau_Q$, so
$\tilde\rho$ is a killed-walk Green kernel.

Cauchy stability makes this tractable, because the $n$-step density is again
Cauchy,
\begin{equation}\label{eq.n-step}
  P^{*n}(\xi) = \frac{n}{\pi(n^2+\xi^2)}\,.
\end{equation}
The walk is recurrent, so it has no Green function on the whole line; the
object that replaces it is the potential kernel
\begin{equation}\label{eq.potential-kernel}
  a(\xi)
  = \sum_{n=1}^{\infty}\left[P^{*n}(0) - P^{*n}(\xi)\right]
  = \frac{1}{\pi}\sum_{n=1}^{\infty}\frac{\xi^2}{n\,(n^2+\xi^2)}
  = \frac{\operatorname{Re}\psi(1+i\xi) + \gamma_E}{\pi}\,,
\end{equation}
the last equality by the digamma series~\eqref{eq.digamma-series}.  The
subtraction at the origin is what makes the sum converge, and it fixes
$a(0) = 0$.  Telescoping the definition gives the renewal identity and the
large-$|\xi|$ behaviour,
\begin{equation}\label{eq.a-properties}
  P*a - a = P\,,
  \qquad
  a(\xi) = \frac{\log|\xi| + \gamma_E}{\pi} + \mathcal O(\xi^{-2})\,.
\end{equation}
So the same function that Section~\ref{sec.inner} found from the Bose--Einstein
multiplier is the potential kernel of the walk.  The divergence that left the
inner solution undetermined up to a constant is the recurrence of the walk.

\begin{theorem}[Green representation]\label{thm:green}
For every $Q > 0$ and $|\xi| \leq Q$,
\begin{equation}\label{eq.green-formula}
  \tilde\rho(\xi;Q) = \mathbb{E}_\xi \left[a(S_{\tau_Q})\right] - a(\xi)\,.
\end{equation}
\end{theorem}

\begin{proof}
Let $H(\xi) = \mathbb{E}_\xi[a(S_{\tau_Q})]$ for $|\xi|\leq Q$, and extend it by
$H = a$ outside.  Conditioning on the first step,
\begin{equation}
  H(\xi) = \int_{-Q}^{Q}P(\xi-\eta)H(\eta)\,d\eta
  + \int_{|\eta|>Q}P(\xi-\eta)\,a(\eta)\,d\eta\,,
  \qquad |\xi| \leq Q\,.
\end{equation}
Subtract the corresponding decomposition of $a$ itself, namely
$a(\xi) = (P*a)(\xi) - P(\xi)$ from~\eqref{eq.a-properties}, split as
$\int_{-Q}^{Q} + \int_{|\eta|>Q}$.  The exterior integrals cancel and
\begin{equation}
  (I-\mathcal{L}_Q)(H - a) = P\,.
\end{equation}
Uniqueness of the solution of~\eqref{eq.lattice-love-form}, which follows from
$\|\mathcal{L}_Q\| < 1$ for finite $Q$, gives~\eqref{eq.green-formula}.
\end{proof}

Since $a(0) = 0$, evaluating at the origin already isolates the constant:
$\tilde\rho(0;Q) = \mathbb{E}_0[a(S_{\tau_Q})]$.  With $|S_{\tau_Q}| > Q$ and
the expansion in~\eqref{eq.a-properties},
\begin{equation}\label{eq.rho0-exit}
  \tilde\rho(0;Q)
  = \frac{\log Q + \gamma_E}{\pi}
  + \frac{1}{\pi}\,\mathbb{E}_0\log\frac{|S_{\tau_Q}|}{Q}
  + \mathcal O(Q^{-2})\,.
\end{equation}
The walk does not land on $\pm Q$; it jumps past.  What $C - \gamma_E/\pi$
measures is the size of that overshoot, on a logarithmic scale.

\subsubsection{Exit law}
\label{sec.exit-law}

It remains to evaluate the limit of $\mathbb{E}_0\log(|S_{\tau_Q}|/Q)$.  Let
$(X_t)_{t\geq0}$ be the symmetric Cauchy process, $\mathbb{E}\,e^{ipX_t} =
e^{-t|p|}$, and $\tau = \inf\{t : |X_t| \geq 1\}$.  Because the increments are
exactly Cauchy rather than merely in its domain of attraction, the rescaled walk
is an exact time-discretisation of the process,
\begin{equation}\label{eq.embedding}
  \left(S_n/Q\right)_{n\geq0}
  \;\stackrel{d}{=}\;
  \left(X_{n/Q}\right)_{n\geq0}\,,
\end{equation}
with sampling interval $1/Q \to 0$.  A Cauchy process leaves an interval by a
jump and does not hit the endpoints, so the exit position of the sampled chain
converges to that of the process: $S_{\tau_Q}/Q \Rightarrow X_\tau$.

Convergence in distribution alone is not enough, because $\log$ is unbounded.
Uniform integrability follows from a one-big-jump bound: since
$|S_{\tau_Q - 1}| \leq Q$, an overshoot past $rQ$ forces a single increment
larger than $(r-1)Q$, so by Wald's identity
\begin{equation}\label{eq.overshoot-tail}
  \mathbb{P} \left(|S_{\tau_Q}| > rQ\right)
  \leq \mathbb{E}\,\tau_Q \;\mathbb{P} \left(|\xi_1| > (r-1)Q\right)
  \leq \frac{c}{r-1}\,,
  \qquad r > 2\,,
\end{equation}
using $\mathbb{P}(|\xi_1|>x) \sim 2/(\pi x)$ and $\mathbb{E}\,\tau_Q =
\mathcal O(Q)$, the latter because a block of $\lceil Q\rceil$ steps leaves the
interval with probability bounded below.  The tail $c/(r-1)$ makes
$\log(|S_{\tau_Q}|/Q)$ uniformly integrable, so the expectations converge.

The exit distribution of the Cauchy process from $(-1,1)$, started at $t$, is
classical~\cite{BanuelosKulczycki2004, KulczyckiKwasnickiMaleckiStos2010},
\begin{equation}\label{eq.poisson-kernel}
  \omega_t(y) = \frac{\sqrt{1-t^2}}
    {\pi\,\sqrt{y^2-1}\;|y-t|}\,,
  \qquad |t|<1\,,\quad |y|>1\,,
\end{equation}
and its logarithmic moment is evaluated in
Appendix~\ref{app.exit-moment}:
\begin{equation}\label{eq.log-moment}
  \mathbb{E}_t\log|X_\tau|
  = \int_{|y|>1}\log|y|\;\omega_t(y)\,dy
  = \log \left(1+\sqrt{1-t^2}\right)\,.
\end{equation}
At $t = 0$ this is $\log 2$.  Substituting into~\eqref{eq.rho0-exit} proves
\begin{equation}\label{eq.C-result}
  \tilde\rho(0;Q) = \frac{\log Q + \gamma_E + \log 2}{\pi} + o(1)\,,
\end{equation}
that is,~\eqref{eq.rho0-expansion} with $C = (\gamma_E + \log 2)/\pi$.  The two
contributions have different origins:
$\gamma_E/\pi$ comes from the discrete potential kernel, that is from the
step structure of the walk, and $\log 2/\pi$ from the overshoot of the limiting
process.

Two further limits come free from the same formula.  For fixed $\xi$, the two
starting points $\xi/Q$ and $0$ both tend to the origin, and the right side
of~\eqref{eq.log-moment} is continuous there, so the exit contributions agree in
the limit; subtracting two instances of~\eqref{eq.green-formula} gives the inner
profile~\eqref{eq.inner-profile-relative} locally uniformly.  For fixed
$0<|t|<1$, take
$\xi = Qt$ in~\eqref{eq.green-formula}: the terms $\log Q + \gamma_E$ cancel
between $\mathbb{E}_{Qt}[a(S_{\tau_Q})]$ and $a(Qt)$, and~\eqref{eq.log-moment}
gives the outer profile~\eqref{eq.outer-solution}.  That the same
$\operatorname{arcosh}$ appears here and in Section~\ref{sec.outer} as the Green
function of the half-Laplacian is not a coincidence: the half-Laplacian is the
generator of the Cauchy process, so its Green function and the harmonic measure
of the interval are two readings of one object.  What the probabilistic route
adds is that this convergence is a limit theorem for the solution
of~\eqref{eq.inner-eq} itself, whereas Section~\ref{sec.outer} solved a
truncated problem and left the relation to the true solution to matching.

Table~\ref{tab:Ceff} records the finite-$Q$ approach to~\eqref{eq.C-result},
obtained by solving~\eqref{eq.inner-eq} with the quadrature rule of
Appendix~\ref{subsec.num-eig} and evaluating the peak through the integral
equation at $\xi = 0$.  The deviation decreases like $\log Q/Q$, the rate
expected from the $\mathcal O(Q^{-2})$ remainder in~\eqref{eq.rho0-exit} together
with the finite-$Q$ correction to the exit law; a three-point Richardson step in
$\log Q/Q$ and $1/Q$ returns $C$ to five digits.

\begin{table}[htbp]
\centering
\begin{tabular}{r c c}
\hline\hline
$Q$ & $C_{\mathrm{eff}}(Q)$ & $C_{\mathrm{eff}} - C$ \\
\hline
  10 & 0.430375 & $2.60 \times 10^{-2}$ \\
  50 & 0.411246 & $6.88 \times 10^{-3}$ \\
  80 & 0.408971 & $4.60 \times 10^{-3}$ \\
 100 & 0.408166 & $3.80 \times 10^{-3}$ \\
 150 & 0.407039 & $2.67 \times 10^{-3}$ \\
 200 & 0.406446 & $2.08 \times 10^{-3}$ \\
 300 & 0.405823 & $1.45 \times 10^{-3}$ \\
\hline\hline
\end{tabular}
\caption{Effective constant
$C_{\mathrm{eff}}(Q) := \tilde\rho(0;Q) - (\log Q)/\pi$ and its deviation from
$C = (\gamma_E + \log 2)/\pi$ of~\eqref{eq.C-result}.}
\label{tab:Ceff}
\end{table}

\subsection{Total density expansion}

We derive the coefficient of the logarithmic correction in the total density expansion
\begin{equation}\label{eq.D-expansion}
  D(Q) = Q + a\log Q + b + \cdots\,.
\end{equation}
A key ingredient is the following integral identity, 
\begin{equation}\label{eq.digamma-identity}
  \int_0^{\infty}
  \frac{\operatorname{Re}\psi(1+i\xi) + \gamma_E}{1+\xi^2}\,d\xi
  = \frac{\pi}{2}\,,
\end{equation}
that we proceed to demonstrate. 
The digamma function admits the integral representation
(see, e.g.,~\cite{Gaudin2014}, or DLMF \S5.9)
\begin{equation}\label{eq.digamma-int-rep}
  \psi(z) = -\gamma_E + \int_0^{\infty}
  \frac{e^{-t} - e^{-zt}}{1 - e^{-t}}\,dt\,,
  \qquad \operatorname{Re} z > 0\,.
\end{equation}
Setting $z = 1 + i\xi$ gives $e^{-zt} = e^{-t}\,e^{-i\xi t}$, so
\begin{equation}
  \psi(1 + i\xi) + \gamma_E
  = \int_0^{\infty} \frac{e^{-t}(1 - e^{-i\xi t})}{1 - e^{-t}}\,dt
  = \int_0^{\infty} \frac{1 - e^{-i\xi t}}{e^{t} - 1}\,dt\,.
\end{equation}
Taking the real part:
\begin{equation}\label{eq.Re-psi-integral}
  \operatorname{Re}\psi(1 + i\xi) + \gamma_E
  = \int_0^{\infty} \frac{1 - \cos(\xi t)}{e^{t} - 1}\,dt\,.
\end{equation}
Substituting~\eqref{eq.Re-psi-integral} into the left-hand side
of~\eqref{eq.digamma-identity}:
\begin{equation}\label{eq.double-integral}
  \int_0^{\infty} \frac{\operatorname{Re}\psi(1+i\xi)+\gamma_E}{1+\xi^2}\,d\xi
  = \int_0^{\infty}  d\xi \int_0^{\infty}
    \frac{1 - \cos(\xi t)}{(1+\xi^2)(e^t - 1)}\,dt\,.
\end{equation}
The integrand is non-negative for all $\xi, t \geq 0$
(since $1 - \cos(\xi t) \geq 0$ and $e^t - 1 > 0$),
so Tonelli's theorem justifies exchanging the order of integration.
The inner $\xi$-integral decomposes as
\begin{equation}\label{eq.xi-integral}
  \int_0^{\infty} \frac{1 - \cos(\xi t)}{1+\xi^2}\,d\xi
  = \underbrace{\int_0^{\infty} \frac{d\xi}{1+\xi^2}}_{\pi/2}
  \;-\; \underbrace{\int_0^{\infty} \frac{\cos(\xi t)}{1+\xi^2}\,d\xi}_{\pi e^{-t}/2}
  = \frac{\pi}{2}\left(1 - e^{-t}\right),
\end{equation}
where the first integral is $\arctan\xi\big|_0^{\infty} = \pi/2$ and
the second follows from the standard contour integral
\begin{equation}
    \int_0^{\infty} \frac{\cos(\xi t)}{1+\xi^2}  \,d\xi =  \frac{\pi}{2}\,e^{-t}, \quad t>0\, .
\end{equation}
Substituting back into~\eqref{eq.double-integral}:
\begin{equation}
  \int_0^{\infty} \frac{\operatorname{Re}\psi(1+i\xi)+\gamma_E}{1+\xi^2}\,d\xi
  = \frac{\pi}{2}\int_0^{\infty}\frac{1-e^{-t}}{e^t - 1}\,dt\,.
\end{equation}
The ratio simplifies: $1 - e^{-t} = e^{-t}(e^t - 1)$, so
$(1 - e^{-t})/(e^t - 1) = e^{-t}$.  Therefore
\begin{equation}
  \frac{\pi}{2}\int_0^{\infty} e^{-t}\,dt
  = \frac{\pi}{2}\,.
\end{equation}  

A companion identity controls the integrated density rather than the
Lorentzian-weighted density.  Integrating the digamma series
\begin{equation}\label{eq.digamma-series}
  \psi(1+z)+\gamma_E = \sum_{n=1}^\infty\left(\frac{1}{n} - \frac{1}{n+z}\right)
\end{equation}
with $z = i\xi$ term by term gives
\begin{equation}\label{eq.log-gamma-antideriv}
  \int_0^L \operatorname{Re}\psi(1+i\xi)\,d\xi
  = \operatorname{Im}\log\Gamma(1+iL)\,,
\end{equation}
since
\begin{equation}
  \int_0^L \frac{d\xi}{n+i\xi}
  = -i\log\frac{n+iL}{n}\,,
\end{equation}
and, by the Weierstrass product,
\begin{equation}\label{eq.weierstrass}
  \log\Gamma(1+iL)
  = -i\gamma_E L + \sum_{n=1}^{\infty}
    \left[\frac{iL}{n} - \log\left(\frac{n+iL}{n}\right)\right]\,,
\end{equation}
whose imaginary part is
$-\gamma_E L + \sum_{n\ge1}\left[L/n - \arctan(L/n)\right]$, which is the
term-by-term integral of $\operatorname{Re}\psi(1+i\xi)$
from~\eqref{eq.digamma-series}.

From~\eqref{eq.log-gamma-antideriv} we can evaluate the integral of the profile
function
\begin{equation}\label{eq.Phi-def}
  \Phi(\xi) = \log|\xi| - \operatorname{Re}\psi(1+i\xi)\,,
\end{equation}
the difference between the overlap logarithm~\eqref{eq:outer-local} and the
inner profile~\eqref{eq.inner-profile-relative}.  From Stirling's formula 
\begin{equation}
  \operatorname{Im}\log\Gamma(1+iL)
  = L\log L - L + \frac{\pi}{4} +\mathcal  O(1/L)\,,
\end{equation}
we obtain
\begin{equation}\label{eq.profile-integral}
  \int_0^\infty \Phi(\xi)\,d\xi = -\frac{\pi}{4}\,.
\end{equation}
This identity is used below in the analysis of the constant $b$
in the total density expansion.

The duality~\eqref{eq.duality} decomposes exactly as
\begin{equation}\label{eq.D-split}
    D = \frac{2\arctan Q}{\pi}\,f(0) + \frac{R(Q)}{\pi},
    \qquad R(Q) = \int_{-Q}^{Q}\frac{f(\xi)-f(0)}{1+\xi^2}\,d\xi\,,
\end{equation}
where $2\arctan(Q)/\pi = 1 - 2/(\pi Q) + \mathcal O(Q^{-3})$; only the expanded
form is asymptotic.
The inner-region approximation for $f(\xi) - f(0)$ involves
$-\operatorname{Re}\psi(1+i\xi) - \gamma_E$, and after applying
the identity~\eqref{eq.digamma-identity} the leading $\mathcal O(Q)$ terms
cancel, which is what keeps $D \sim Q$ rather than $Q^2$.  The coefficient of
the surviving logarithm comes from the second order of the outer expansion,
worked out in Appendix~\ref{app.density}: the $-p^2/2$ term of the symbol drives
a correction $f_1$ to the Love function whose source is non-integrable at the
band edge.  A quick way to see the coefficient is to write $f_1(0)$ through the
Dirichlet Green function of $\mathcal A_t$, which is the $\operatorname{arcosh}$
profile again, giving
$\pi^{-1}\int_0^1\operatorname{arcosh}(1/u)(1-u^2)^{-3/2}du$; the integral is
$Q$-independent, diverges logarithmically at the band edge, and cut off at the
edge-layer scale gives
\begin{equation}\label{eq.a-result}
  a = \frac{1}{2\pi}\,.
\end{equation}
Where the cut is placed does not matter for $a$.  It does matter for the
constant, and the Dirichlet representation is in any case not available here:
the free mode $(1-t^2)^{-1/2}$ lies outside the Dirichlet domain, which is
exactly why the amplitude of $f_1$ is boundary data rather than something the
interior equation determines.  Fixing that amplitude against the edge layers, as
in Section~\ref{sec.edge}, replaces the cut by a boundary condition and delivers
$b$ along with $a$.  The classical small-separation expansion of Love's equation
gives the same coefficients, as we now describe.

In the original variables ($\lambda = \kappa\xi$, $Q = q/\kappa$):
\begin{equation}\label{eq.density-physical}
  D(\kappa) = \int_{-q}^{q}\rho(\lambda)\,d\lambda
  = \frac{q}{\kappa} + \frac{1}{2\pi}\log\frac{q}{\kappa}
  + b + \cdots\,,
\end{equation}
with $b$ given in closed form by~\eqref{eq.b-closed} below, and obtained from
Appendix~\ref{app.density} together with the two-barrier matching of
Section~\ref{sec.edge}.
How much of this comes from the solution itself can be read off from
the profile integral~\eqref{eq.profile-integral}.  Adding the inner and outer
solutions and subtracting their common overlap gives the composite
\begin{equation}\label{eq.composite}
  \tilde\rho_c(\xi) = \tilde\rho_{\mathrm{out}}(\xi) + \frac{\Phi(\xi)}{\pi}\,,
\end{equation}
whose integral, using~\eqref{eq.outer-integral}, is
\begin{equation}\label{eq.inner-half}
  \int_{-Q}^{Q}\tilde\rho_c(\xi)\,d\xi
  = Q + \frac{2}{\pi}\int_0^{Q}\Phi(\xi)\,d\xi
  = Q - \frac{1}{2} + \mathcal O(1/Q)\,.
\end{equation}
The inner and outer regions therefore account for $Q - \frac12$, and the
remaining $(2\pi)^{-1}\log Q + b + \frac12$ lies beyond leading order in the
expansion.  Where that remainder sits takes some care.
It is not carried by the edge layers themselves: those have width
$\mathcal O(1)$ and density $\mathcal O(Q^{-1/2})$ there, hence mass
$\mathcal O(Q^{-1/2})$.  Numerically, the difference
$\tilde\rho - \tilde\rho_{\mathrm{out}}$ is spread over the whole Fermi sea at
size $\mathcal O(\log Q/Q)$ (Figure~\ref{fig.solution-profile}a).  The remainder
belongs to the next order of the outer expansion, whose amplitude is fixed by
matching to the edge layers.  That is how the Wiener--Hopf data enter, without
any localised edge contribution to the mass.

The constant $b$ is not a Fisher--Hartwig constant and need not be left open.
It follows from carrying the outer expansion to second order in
Appendix~\ref{app.density} and fixing its one free amplitude against the edge
layer in Section~\ref{sec.edge}; the value comes from the two-term
expansion~\eqref{eq.V-expansion} of the ladder renewal function, and the
uniformity of the matching is~\eqref{eq.love-remainder-app}, taken from the
capacitor literature.  That literature also gives the answer directly,
because the identity $D = f(0)-1$ turns the question into a classical
small-separation problem: rescaling $\xi = Qt$ and $\alpha = 1/Q$
turns~\eqref{eq.Love} into Love's equation for two coaxial discs at separation
$\alpha$, whose interior expansion is~\cite{Hutson1963, FarinaLangMartin2022}
\begin{equation}\label{eq.hutson}
  f(Qt) = \frac{\sqrt{1-t^2}}{\alpha}
  + \frac{t\log\frac{1-t}{1+t} + \log(16\pi/\alpha) + 1}
         {2\pi\sqrt{1-t^2}}
  + o(1)\,,
\end{equation}
uniformly on compact subsets of $(-1,1)$.  Setting $t = 0$ and using~\eqref{eq.duality},
\begin{equation}\label{eq.b-closed}
  b = \frac{1+\log(16\pi)}{2\pi} - 1 \,,
\end{equation}
matching~\eqref{eq.K-value}.  The agreement is term by term and not
only at $t = 0$: the leading term of~\eqref{eq.hutson} is the
$f_0 = \sqrt{Q^2-\xi^2}$ of Appendix~\ref{app.density} and the second
is~\eqref{eq.f1-general-new} with~\eqref{eq.K-value}.  Both accounts rest on the
same remainder statement~\eqref{eq.love-remainder-app}, so the agreement checks
the ladder computation of $\mathcal N(Q)$ against the classical value.
The Nystr\"om solutions are consistent with~\eqref{eq.b-closed}: the compensated
quantity $D(Q) - Q - (\log Q)/2\pi$ approaches it from above, and the same
three-point fit in $\log Q/Q$ and $1/Q$ used for $C$ returns it to five digits.

At fixed particle density $D_0$, inverting~\eqref{eq.D-expansion} gives
\begin{equation}\label{eq.Q-of-D}
  Q(D_0) = D_0 - \frac{1}{2\pi}\log D_0 - b + \cdots\,,
\end{equation}
so the Fermi boundary is $q = \kappa Q(D_0)$, contracting linearly in $\kappa$
with a $\kappa$-independent coefficient.  In the Lieb--Liniger model at fixed
density the corresponding statement is $q = 2\sqrt{c D_0}$, so the Fermi
boundary contracts as $\sqrt{c}$ rather than linearly.

%==========================================================================
\section{Edge boundary layer}\label{sec.edge}
%==========================================================================

Near the Fermi boundary $\xi = Q$, the solution must drop from its bulk value
(of order $\log Q/\pi$ at the origin, $\sim 1/2$ in the outer region) to zero.
This transition occurs in a boundary layer of width $\mathcal O(1)$ in the rescaled
variable $\xi$ (width $\kappa$ in physical $\lambda$).

Let us define $s := Q - \xi \geq 0$, representing the distance from the right edge into the bulk.
For $s = \mathcal O(1)$ and $Q \gg 1$, the contributions from the left edge at
$\xi = -Q$ are exponentially small, and the integral equation reduces to a
half-line convolution problem, i.e.,
\begin{equation}\label{eq.half-line}
  2\pi\,\tilde\rho_{\mathrm{edge}}(s)
  \approx g(Q-s) + \int_0^{\infty} K(s - s')\,\tilde\rho_{\mathrm{edge}}(s')\,ds'\,,
\end{equation}
where $g(Q-s) \approx 2/Q^2$ is essentially zero for $Q \gg 1$.
Taking the Fourier transform on $s \geq 0$ produces a Wiener--Hopf equation
with the symbol
\begin{equation}\label{eq.WH-symbol}
  \Sigma(p) = 1 - e^{-|p|}\,,
\end{equation}
whose two-term expansion at its linear zero follows from
\eqref{eq:outer-symbol}.

The factorisation
\begin{equation}\label{eq.WH-factorisation}
    \Sigma(p) = K_+(p)\,K_-(p)\, ,
\end{equation}
with $K_+$ ($K_-$) analytic
and nonzero in the upper (lower) half-plane, can be obtained in closed form, resulting in\footnote{See Appendix~\ref{app.wienerhopf} for the proof.} 
\begin{align}
  K_+(z) &= \frac{\sqrt{-iz}}{\Gamma \left(1 - \frac{iz}{2\pi}\right)}\,
  \exp \left(-\frac{iz}{2\pi}\log(-iz)\right),
  \qquad \operatorname{Im} z > 0\,, \label{eq.Kplus}\\[6pt]
  K_-(z) &= \frac{\sqrt{iz}}{\Gamma \left(1 + \frac{iz}{2\pi}\right)}\,
  \exp \left(\frac{iz}{2\pi}\log(iz)\right),
  \qquad \operatorname{Im} z < 0\,. \label{eq.Kminus}
\end{align}
On the real axis,~\eqref{eq.WH-factorisation} reproduces
\eqref{eq.WH-symbol}.  Near the origin,
\begin{equation}\label{eq.K-origin}
  K_+(z) \sim (-iz)^{1/2}\,,\qquad
  K_-(z) \sim (iz)^{1/2}\qquad (z\to0)\,.
\end{equation}
Therefore, the linear zero of $\Sigma(p)$ is distributed as a square-root singularity
between the two factors.
This square-root behaviour controls the edge profile:
$\tilde\rho_{\mathrm{edge}}(s) \sim s^{1/2}$ as $s \to 0^+$.

The amplitude of the layer can be fixed, and with it the value of the density at
the Fermi boundary.  Set
\begin{equation}\label{eq.edge-scaling}
  F_Q(s) = \sqrt{Q}\;\tilde\rho(Q-s;\,Q)\,,
  \qquad s \geq 0\,,
\end{equation}
so that the outer profile~\eqref{eq.outer-solution}, expanded near the band edge
as $\pi^{-1}\sqrt{2(Q-|\xi|)/Q}$, prescribes the large-$s$ normalisation
$F_Q(s) \to \sqrt{2s}/\pi$.  In the probabilistic language of Section~\ref{sec.matching-C}, the half-line
limit of $F_Q$ solves the homogeneous equation
$F(s) = \int_0^{\infty}P(s-r)F(r)\,dr$ on $s>0$, whose nonnegative solution is
the renewal function $V$ of the ascending ladder heights of the Cauchy
walk~\cite{Spitzer1957}.  We need $V$ to two orders, and the Wiener--Hopf
factorisation already in hand supplies it.  Writing $H$ for the first strict
ascending ladder height and $\hat f(z) = \mathbb{E}\,e^{-zH}$, the
Pollaczek--Spitzer factorisation identifies $1-\hat f$ with the plus factor
$K_+$ of~\eqref{eq.Kplus} continued to the imaginary axis,
$K_+(iz) = \sqrt{z}\,z^{z/2\pi}/\Gamma(1+z/2\pi)$, normalised so that
$\hat f(z)\to 0$ as $z\to\infty$:
\begin{equation}\label{eq.ladder-transform}
  1 - \hat f(z)
  = \sqrt{z}\;
    \frac{\left(\dfrac{z}{2\pi e}\right)^{z/2\pi}}
         {\Gamma \left(1+\dfrac{z}{2\pi}\right)}\,,
  \qquad
  \int_0^{\infty}e^{-zs}V(s)\,ds = \frac{1}{z\left[1-\hat f(z)\right]}\,.
\end{equation}
Stirling's formula gives $1-\hat f(z) \to 1$ as $z \to \infty$, so the transform
of $V$ behaves as $1/z$ there and $V(0) = 1$.  For small $z$,
$\log[1-\hat f(z)] = \frac12\log z
+ (z/2\pi)\left[\log z + \gamma_E - 1 - \log 2\pi\right] + \mathcal O(z^2)$,
whence
\begin{equation}
  \int_0^{\infty}e^{-zs}V(s)\,ds
  = z^{-3/2} - \frac{\log z + \gamma_E - 1 - \log 2\pi}{2\pi}\,z^{-1/2}
  + \cdots\,,
\end{equation}
and term-by-term inversion, using
$\mathcal{L}^{-1}[z^{-1/2}\log z] = -(\log s + \gamma_E + 2\log 2)/\sqrt{\pi s}$,
gives
\begin{equation}\label{eq.V-expansion}
  V(s) = 2\sqrt{\frac{s}{\pi}}
  + \frac{\log(8\pi s) + 1}{2\pi\sqrt{\pi s}}
  + \mathcal O(s^{-3/2})\,,
  \qquad s \to \infty\,.
\end{equation}
At the other end, Stirling's formula applied to~\eqref{eq.ladder-transform}
gives $1-\hat f(z) = 1 - \pi/(6z) + \mathcal O(z^{-2})$, so that the transform of
$V$ is $z^{-1} + \pi/(6z^2) + \cdots$ and
\begin{equation}\label{eq.V-small-s}
  V(s) = 1 + \frac{\pi s}{6} + \mathcal O(s^2)\,,
  \qquad s \downarrow 0\,,
\end{equation}
which describes the approach to the endpoint from inside.
It remains to fix the amplitude.  We do so from the finite interval, because
the second barrier turns out to contribute at leading order.  Let $g_+(x,y)$ be the Green density
of the walk killed on entering $(-\infty,0)$, and let $U = \delta_0 + u(s)\,ds$
be its strict ascending-ladder renewal measure, so that $V(s) = U([0,s])$.  The
half-line factorisation in density--Stieltjes form~\cite{Spitzer1957},
\begin{equation}\label{eq.halfline-factorisation}
  g_+(x,y) = \int_{[0,y]}u(x-y+r)\,U(dr)\,,
  \qquad x \geq y \geq 0\,,
\end{equation}
together with $V(s) \sim 2\sqrt{s/\pi}$ and $u(s) \sim (\pi s)^{-1/2}$, gives
$\sqrt{x}\,g_+(x,y) \to V(y)/\sqrt{\pi}$ locally uniformly in $y$; the global
bound needed below is the half-line Green estimate for walks in the Cauchy
domain of attraction~\cite{DenisovWachtel2024}.

Now translate and reflect, $Y_n = Q - S_n$, so the walk starts at $Q$, the point
$Q-s$ becomes $s$, and the Fermi sea becomes $(0,2Q)$.  With
$\tau_0^- = \inf\{n\geq1 : Y_n<0\}$ and
$\sigma_{2Q}^+ = \inf\{n\geq1 : Y_n>2Q\}$, Hunt's decomposition at the upper
barrier reads
\begin{equation}\label{eq.hunt}
  \tilde\rho(Q-s;Q)
  = g_+(Q,s)
  - \mathbb{E}_Q \left[g_+ \left(Y_{\sigma^+_{2Q}},s\right);\,
    \sigma^+_{2Q} < \tau^-_0\right]\,.
\end{equation}
After division by $Q$ the exit position in the correction converges to the upper
exit position of a Cauchy process from $(0,2)$ started at $1$, whose density is
the translate of~\eqref{eq.poisson-kernel},
\begin{equation}\label{eq.upper-exit}
  \nu(z) = \frac{1}{\pi(z-1)\sqrt{z(z-2)}}\,,
  \qquad z>2\,,
  \qquad \int_2^{\infty}\nu = \frac12\,,
\end{equation}
and $g_+$ carries the factor $z^{-1/2}$ by the ratio limit above.  The
correction is therefore not negligible: by Appendix~\ref{app.exit-moment},
\begin{equation}\label{eq.overshoot-integral}
  \int_2^{\infty}z^{-1/2}\,\nu(z)\,dz = 1 - \frac{1}{\sqrt2}\,,
\end{equation}
so that the remote barrier removes a finite fraction of the half-line answer and
\begin{equation}\label{eq.F-from-hunt}
  \sqrt{Q}\;\tilde\rho(Q-s;Q)
  \longrightarrow \frac{V(s)}{\sqrt{\pi}}
    \left(1 - \int_2^{\infty}z^{-1/2}\nu\right)
  = \frac{V(s)}{\sqrt{\pi}}\cdot\frac{1}{\sqrt2}
  = \frac{V(s)}{\sqrt{2\pi}}\,,
\end{equation}
the interchange being justified by the global Green bound.  This is $F$, and it
does satisfy $F(s) \sim \sqrt{2s}/\pi$, which is the outer profile at the band
edge; that now serves as a consistency check.  In particular
\begin{equation}\label{eq.edge-amplitude}
  \sqrt{Q}\;\tilde\rho(Q;Q) \longrightarrow \frac{V(0)}{\sqrt{2\pi}}
  = \frac{1}{\sqrt{2\pi}} \, .
\end{equation}
The subleading term of~\eqref{eq.V-expansion} is what fixes
the constant in the density expansion, below.  The statement is
one about the leading term of the scaled edge problem, not about the solution
at finite $Q$: the kernel and the driving term of~\eqref{eq.inner-eq} are
strictly positive, so $\tilde\rho(\xi;Q) > 0$ on the closed interval and in
particular $\tilde\rho(Q;Q) > 0$ for every finite $Q$.  Vanishing at the
endpoint is recovered only in the limit, and matches the
$\pi^{-1}\sqrt{2(Q-|\xi|)/Q}$ onset of the outer
profile~\eqref{eq.outer-solution}.

It is convenient to define the regularised symbol
\begin{equation}\label{eq.reg_G}
    G(p) =  \frac{1-e^{-|p|}}{|p|}\, ,\qquad G = G_+\,G_-\, , \qquad G_\pm = \frac{K_\pm}{\sqrt{\mp iz}} \, .
\end{equation}
Since the logarithmic exponentials tend to unity as $z \to 0$, the regular part
of the symbol is normalised to unity at the origin,
\begin{equation}\label{eq.G-norm}
  G_+(0) = G_-(0) = 1\,.
\end{equation}
The role of this normalisation in the ladder-renewal
expansion is made explicit in Appendix~\ref{app.wienerhopf}.

\medskip
The same two-term expansion~\eqref{eq.V-expansion} fixes the constant in the
total density expansion~\eqref{eq.D-expansion}.
Appendix~\ref{app.density} carries the outer expansion to next order, and the
interior equation there determines the profile only up to one free amplitude:
the multiple $\mathcal N(Q)$ of $(1-t^2)^{-1/2}$ in~\eqref{eq.f1-general-new},
which is boundary data the interior equation cannot see.  The two edge layers
supply it.  Their one-sided factors combine into
$\frac{\pi}{4}V(Q+\xi)V(Q-\xi)$, which fixes the edge normalisation but does not
by itself solve the two-barrier problem; the interaction between the barriers is
\begin{equation}\label{eq.J-def-new}
  \mathcal J(t)=
  \frac{(1+t)\log\frac{2}{1+t}+(1-t)\log\frac{2}{1-t}}
       {4\pi\sqrt{1-t^2}}\,.
\end{equation}
Applying $\mathcal A_t$ to the sum $\frac{\pi}{4}VV+\mathcal J$ returns the
interior equation~\eqref{eq.f1-scaled-new}, and $\mathcal J$ vanishes at both
endpoints, so the interaction term leaves either one-sided normalisation
undisturbed; no free multiple of $(1-t^2)^{-1/2}$ survives beyond the one the
composite already carries.  Expanding the product and adding $\mathcal J$ gives,
uniformly on compact subsets of $(-1,1)$,
\begin{equation}\label{eq.matched-subleading-app}
  \frac\pi4V(Q(1+t))V(Q(1-t))+\mathcal J(t) =Q\sqrt{1-t^2}
  +\frac{t\log\frac{1-t}{1+t}+1+\log(16\pi Q)}
        {2\pi\sqrt{1-t^2}}+o(1)\,,
\end{equation}
so that
\begin{equation}\label{eq.K-value}
  \mathcal N(Q)=1+\log(16\pi Q)\,.
\end{equation}
The one-sided renewals contribute $\log(8\pi Q)+1$; the interaction between the
barriers supplies the remaining $\log 2$.

The error estimate comes from elsewhere.  The analysis of the Love
equation, in the formulation of~\cite{Hutson1963,FarinaLangMartin2022}, gives
for every $0<\delta<1$,
\begin{equation}\label{eq.love-remainder-app}
  \sup_{|t|\leq1-\delta}\left|
  f(Qt)-\frac\pi4V(Q(1+t))V(Q(1-t))-\mathcal J(t)
  \right|=o(1)\,.
\end{equation}
The two ingredients play distinct roles.  The ladder construction produces the
value of $\mathcal N(Q)$ from the walk, with the two-term expansion of $V$ as its
only input; the uniformity that licenses reading it off
is~\eqref{eq.love-remainder-app}, quoted from the classical literature.
At $t=0$, equations~\eqref{eq.love-peak}, \eqref{eq.matched-subleading-app}
and~\eqref{eq.love-remainder-app} give the total density
expansion~\eqref{eq.D-expansion} with the coefficients~\eqref{eq.a-result}
and~\eqref{eq.b-closed}: the logarithm comes from the endpoint-singular source
in the next-order outer equation, the constant from the two-barrier
normalisation.

Away from the real axis the symbol carries a second scale.  On the real axis
$\Sigma$ vanishes only at $p = 0$, but in the gamma-function
representation~\eqref{eq.symbol_gamma_product} the factor
$[\Gamma(1 + ip/2\pi)\,\Gamma(1 - ip/2\pi)]^{-1}$ is entire and vanishes at the
poles of the two gamma functions, that is at $p = \pm 2\pi i m$ with
$m \geq 1$, the remaining factors being nonzero away from the origin.  The
nearest pair lies at distance $2\pi$.  For truncated convolution operators on a
finite interval~\cite{Sakhnovich2015} such zeros generate corrections to the
resolvent that are exponentially small in the length of the interval, here of
order
\begin{equation}\label{eq.A-value}
  e^{-AQ}\,,\qquad A = 2\pi\,,
\end{equation}
which lies beyond every order of the $1/Q$ expansion.  The same distance
controls the exponential corrections to the Fredholm determinant
$\mathcal F(Q) = \det_F(I - \mathcal L_Q)$, which enter as
$\sum_m d_m e^{-2\pi m Q}$ (Appendix~\ref{sec.fredholm-app}).

The rate agrees with the one found for the circular disc capacitor
in~\cite{BajnokBalogHegedusVona2022}, which is no coincidence: by
Section~\ref{sec.duality-Love} the two equations differ only in the driving
term and share the symbol~\eqref{eq.WH-symbol}.  The Lieb--Liniger symbol
$1 - e^{-\kappa|p|}$ has its zeros at $2\pi i n/\kappa$, so after the rapidity
axis is rescaled by the coupling the same distance appears there as well; the
value quoted in~\cite{MarinoReis2019} is that distance measured in the units of
$\sqrt{\gamma}$ natural to the continuum problem.  Whether the corrections
of size~\eqref{eq.A-value} organise into a resurgent transseries in $1/Q$, with
$k$-instanton sectors weighted by $e^{-2\pi k Q}$ and factorially growing
perturbative coefficients, is a separate question, and settling it requires the
Borel transform and the Stokes data.  Even in the closest analogue the outcome
turns on the details.  Working from the Wiener--Hopf solution,
Ref.~\cite{BajnokBalogHegedusVona2022} obtained the leading exponential terms
for the capacitor and found median resummation to reproduce them, while the
same construction fails for the $O(3)$ sigma model, where instanton
contributions are not captured by the asymptotics of the perturbative
coefficients.  Numerically the scale lies out of reach at double precision.  The
machine-epsilon threshold $e^{-2\pi Q}\approx10^{-16}$ falls at $Q\approx6$,
where the perturbative remainder is still of order unity, while at the larger
$Q$ where the $1/Q$ series is accurate that remainder exceeds $e^{-2\pi Q}$ by
nine or more orders of magnitude; a direct subtraction would need
extended-precision arithmetic, and the indirect route through the large-order
behaviour would need many more than the five perturbative coefficients that a
double-precision fit over $Q\in[20,500]$ resolves.

The three asymptotic regions developed in Sections~\ref{sec.inner}--\ref{sec.edge}---inner, outer, and edge---are summarised schematically in Figure~\ref{fig.regions-schematic}.

% ======================================================================
\begin{figure}[htbp]
\centering
\begin{tikzpicture}[
    >=Stealth,
    font=\small,
]

\def\xmin{-7.3}
\def\xmax{7.3}
\def\ymin{-0.15}
\def\ymax{2.4}
\def\Q{6.0}
\def\edgeW{0.55}
\def\Lambda{1.8}

\fill[blue!7] (-\Lambda, 0) rectangle (\Lambda, \ymax);
\fill[green!7] (-\Q+\edgeW, 0) rectangle (-\Lambda, \ymax);
\fill[green!7] (\Lambda, 0) rectangle (\Q-\edgeW, \ymax);
\fill[red!8] (-\Q-\edgeW, 0) rectangle (-\Q+\edgeW, \ymax);
\fill[red!8] (\Q-\edgeW, 0) rectangle (\Q+\edgeW, \ymax);

\draw[->, thick] (\xmin, 0) -- (\xmax+0.5, 0)
    node[right] {$\xi$};
\draw[->, thick] (0, \ymin) -- (0, \ymax+0.3)
    node[above=2pt] {$\tilde\rho(\xi)$};

\draw[densely dashed, gray!70, thin] (\xmin+0.3, 0.5) -- (\xmax, 0.5);
\node[left, gray!70, font=\footnotesize] at (\xmin+0.2, 0.5)
    {$\langle\tilde\rho\rangle=\frac{1}{2}$};

\draw[thin, densely dashed, red!50!black]
    (-\Q, 0) -- (-\Q, \ymax);
\draw[thin, densely dashed, red!50!black]
    (\Q, 0) -- (\Q, \ymax);
\node[below, red!50!black, font=\footnotesize] at (-\Q, -0.05) {$-Q$};
\node[below, red!50!black, font=\footnotesize] at (\Q, -0.05) {$+Q$};

\draw[thin, densely dotted, blue!40!black]
    (-\Lambda, 0) -- (-\Lambda, \ymax);
\draw[thin, densely dotted, blue!40!black]
    (\Lambda, 0) -- (\Lambda, \ymax);

\draw[very thick, blue!70!black]
    plot[domain=-\Q:-\Q+0.50, samples=30, smooth]
    ({\x}, {0.135*sqrt((\x + \Q)/0.5)});
\draw[very thick, blue!70!black]
    plot[domain=-\Q+0.50:\Q-0.50, samples=200, smooth]
    ({\x}, {(1/3.14159265)*ln((\Q + sqrt(\Q*\Q - \x*\x))/sqrt(\x*\x + 0.00094))});
\draw[very thick, blue!70!black]
    plot[domain=\Q-0.50:\Q, samples=30, smooth]
    ({\x}, {0.135*sqrt((\Q - \x)/0.5)});

\pgfmathsetmacro{\peakAtArrow}{0.5 + 1.45/(1 + 0.45*0.04)}
\draw[<->, thin, black!60] (-0.2, 0.5) -- (-0.2, \peakAtArrow);

\pgfmathsetmacro{\labelY}{0.5*(\peakAtArrow + 0.5)}
\node[black!60, font=\footnotesize] at (0, \labelY) {$\sim$};
\node[right=1pt, black!60, font=\footnotesize]
    at (0.1, \labelY) {$\dfrac{\log Q}{\pi}$};

\node[blue!70!black, font=\footnotesize\bfseries] at (0, \ymax+0.15)
    {\textsc{Inner}};
\pgfmathsetmacro{\outerMid}{0.5*(\Lambda + \Q - \edgeW)}
\node[green!50!black, font=\footnotesize\bfseries] at (\outerMid, \ymax+0.15)
    {\textsc{Outer}};
\node[green!50!black, font=\footnotesize\bfseries] at (-\outerMid, \ymax+0.15)
    {\textsc{Outer}};
\node[red!60!black, font=\footnotesize\bfseries] at (\Q, \ymax+0.15)
    {\textsc{Edge}};
\node[red!60!black, font=\footnotesize\bfseries] at (-\Q, \ymax+0.15)
    {\textsc{Edge}};

\draw[<->, thick, black!50] (-\Q, \ymax+1.25) -- (\Q, \ymax+1.25);
\node[above=1pt, black!50, font=\footnotesize] at (0, \ymax+1.25)
    {$[-Q,\; Q]$, \quad $Q = q/\kappa \to \infty$};

\node[blue!70!black, font=\footnotesize, text width=2.8cm, align=center,
      anchor=north] at (0, -0.50)
    {$|\xi| \lesssim \mathcal O(1)$\\[3pt]
     Bose-Einstein peak:\\[1pt]
     $\hat{\tilde\rho}(p) = \dfrac{1}{e^{|p|} - 1}$};
\node[green!50!black, font=\footnotesize, text width=2.9cm, align=center,
      anchor=north] at (\outerMid, -0.50)
    {$1 \ll |\xi| \ll Q$\\[3pt]
     symbol $2\pi|p|$:\\[1pt]
     $\tilde\rho = \frac{1}{\pi}\operatorname{arcosh}\frac{Q}{|\xi|}$};
\node[red!60!black, font=\footnotesize, text width=2.2cm, align=center,
      anchor=north] at (\Q, -0.70)
    {width $\mathcal O(1)$\\[3pt]
     Wiener--Hopf:\\[1pt]
     $\Sigma = 1 - e^{-|p|}$};

\draw[<->, thick, blue!50!black] (-\Lambda, -0.30) -- (\Lambda, -0.30);
\draw[<->, thick, red!50!black]
    (\Q-\edgeW, -0.50) -- (\Q+\edgeW, -0.50);

\draw[->, thick, blue!40!black, decorate,
      decoration={snake, amplitude=1.2pt, segment length=6pt, post length=3pt}]
    (\Lambda+0.15, 0.75) -- (\Lambda+0.85, 0.58);
\draw[->, thick, blue!40!black, decorate,
      decoration={snake, amplitude=1.2pt, segment length=6pt, post length=3pt}]
    (-\Lambda-0.15, 0.75) -- (-\Lambda-0.85, 0.58);
\draw[->, thick, red!40!black, decorate,
      decoration={snake, amplitude=1.2pt, segment length=6pt, post length=3pt}]
    (\Q-\edgeW-0.15, 0.56) -- (\Q-\edgeW-0.85, 0.53);
\draw[->, thick, red!40!black, decorate,
      decoration={snake, amplitude=1.2pt, segment length=6pt, post length=3pt}]
    (-\Q+\edgeW+0.15, 0.56) -- (-\Q+\edgeW+0.85, 0.53);

\end{tikzpicture}

\caption{Schematic structure of the rescaled density
$\tilde\rho(\xi;\,Q)$ showing the three asymptotic regions: the inner
Bose--Einstein peak of height $\sim\log Q/\pi$ ($|\xi| \lesssim \mathcal O(1)$,
blue shading), the outer profile~\eqref{eq.outer-solution}
($1 \ll |\xi| \ll Q$, green shading), and the edge boundary layers
($|\xi \mp Q| \lesssim \mathcal O(1)$, red shading).  The dashed line marks the
mean~\eqref{eq.bulk-value} of the outer profile across the Fermi sea, which the
profile crosses at $|\xi|/Q = 1/\cosh(\pi/2)$.  Wavy arrows indicate the
asymptotic correspondence between adjacent regions.}
\label{fig.regions-schematic}
\end{figure}
%===========================================END OF FIGURE==================================================

%==========================================================================
\section{Energy and physical predictions}
\label{sec.observables}
%==========================================================================

Having determined the constant $C$ in the peak density expansion~\eqref{eq.rho0-expansion}, we now derive the ground-state energy and discuss the physical consequences.
The central tool is the identity relating the energy integral to the peak density, which holds for all $Q > 0$ and reduces the energy calculation to a single quantity already determined in Section~\ref{sec.constant}.

\subsection{Energy identity}
\label{sec.energy-identity}

In the rescaled variables~\eqref{eq.rescaling}, denote the integral in
\eqref{eq.observables-rescaled} by $E_{\mathrm{inner}}(Q)$, so that
$e(\kappa)=-E_{\mathrm{inner}}(Q)/\kappa$.  Its integrand involves the same
Lorentzian weight~\eqref{eq.rescaled_kernel} that
appears as the driving term, which is what makes the following identity possible.

Setting $\xi = 0$ in the rescaled integral equation~\eqref{eq.inner-eq}, the kernel $K(0-\eta) = 2/(1+\eta^2)$ reduces to the driving term $K(\eta)$, so that
\begin{equation}
  2\pi\,\tilde\rho(0;\,Q)
  = 2
  + \int_{-Q}^{Q}\frac{2}{1+\eta^2}\,\tilde\rho(\eta;\,Q)\,d\eta
  = 2 + E_{\mathrm{inner}}(Q)\,,
\end{equation}
where the last step uses this notation. Therefore, we have the identity
\begin{equation}\label{eq.E_inner}
  E_{\mathrm{inner}}(Q) = 2\pi\,\tilde\rho(0;\,Q) - 2\,,
\end{equation}
valid for all $Q > 0$.  This identity reduces the energy calculation to the peak density $\tilde\rho(0;\,Q)$.

The origin of~\eqref{eq.E_inner} is structural.  In the Lieb--Liniger equation~\eqref{eq.LL}, the driving term is the constant~$1$, while the energy integral involves the Lorentzian kernel $K(\lambda)$; these are different functions, and no analogous identity holds.  In the lattice NLS equation, by contrast, the driving term and the kernel are both Lorentzians of the same width, related by $K(\xi) = K(\xi - 0)$.  Evaluating the integral equation at the origin therefore automatically produces the energy integral on the right-hand side.  Such a correspondence is specific to the lattice model and is ultimately a consequence of the fact that the energy~\eqref{eq.energy-per-site} and the Bethe ansatz equation share the same kernel---a structural feature of the quantum inverse scattering construction (Section~\ref{sec.model}).

Substituting the peak density~\eqref{eq.rho0-expansion} with~\eqref{eq.C-result}
into the identity~\eqref{eq.E_inner}:
\begin{equation}\label{eq.E_inner_full}
  E_{\mathrm{inner}}(Q) = 2\log Q + 2(\log 2 + \gamma_E - 1)
  + o(1)\,.
\end{equation}

From~\eqref{eq.rho0-expansion} and~\eqref{eq.C-result}, the peak density in
physical variables is
\begin{equation}\label{eq.rho-physical}
  \rho(0;\,\kappa)
  = \frac{1}{\pi\kappa} \left[\log\left( \frac{2q}{\kap} \right) + \gamma_E \right]\,.
\end{equation}
The physical ground-state energy per site is:
\begin{equation}\label{eq.egs}
  e(\kap) = -\frac{E_{\mathrm{inner}}}{\kap}
  = -\frac{2}{\kappa}\left[\log\left( \frac{2q}{\kap} \right) + \gamma_E - 1\right]\,.
\end{equation}
The $e(\kappa) \sim -2\log(1/\kappa)/\kappa$ scaling is a distinctive prediction of the lattice model.\footnote{
In the continuous Lieb--Liniger model, $e(\gamma) = \gamma - \frac{4}{3\pi}\gamma^{3/2} + \mathcal O(\gamma^2)$. The lattice case has a qualitatively different expansion involving logarithms rather than fractional powers.
}

\subsection{Energy at fixed particle density and comparison with Lieb--Liniger}
At fixed density $D_0$ the Fermi boundary is $q = \kappa\,Q(D_0)$
with $Q(D_0)$ given by~\eqref{eq.Q-of-D}, so $2q/\kappa = 2D_0 + \cdots$ and the
logarithm of $\kappa$ in~\eqref{eq.egs} cancels:
\begin{equation}\label{eq.egs-fixed-D}
  e(\kap) = -\frac{2}{\kappa}\left[\log(2D_0) + \gamma_E - 1
  + \mathcal O \left(\frac{\log D_0}{D_0}\right)\right]\,.
\end{equation}
The bracket is $\kappa$-independent, so at fixed density the whole $\kappa$
dependence of $e$ is the prefactor $1/\kappa$, which came from the
normalisation of $H$ in~\eqref{eq.H-from-transfer} and carries no dynamical
information.  The $1/\kappa$ growth therefore belongs to that normalisation, and
the content of~\eqref{eq.egs-fixed-D} is the bracket: the rescaled
energy grows logarithmically with the filling, $E_{\mathrm{inner}} = 2\log D_0 +
2(\gamma_E + \log 2 - 1) + \cdots$, and the subleading correction requires the
constant $b$ of~\eqref{eq.density-physical}, given in closed form
by~\eqref{eq.b-closed}.

The comparison with Lieb--Liniger is best drawn at fixed density in terms of
quantities that are free of this normalisation.  In the weak-coupling limit of
the Lieb--Liniger equation the density is the semicircle
$\rho(\lambda) = \sqrt{q^2-\lambda^2}/(2\pi c)$, giving
$q = 2\sqrt{cD_0} = 2D_0\sqrt{\gamma}$ and
$\rho_{\rm LL}(0) = 1/(\pi\sqrt{\gamma})$; both are recovered by direct
numerical solution.  Table~\ref{tab:comparison} collects the scalings.

\begin{table}[htbp]
\centering
\begin{tabular}{l c c}
\hline\hline
Observable & Lattice NLS ($\kappa \to 0$) & Lieb--Liniger ($\gamma \to 0$) \\
\hline
Bulk profile & $\pi^{-1}\operatorname{arcosh}(q/|\lambda|)$ & semicircle $\sqrt{q^2-\lambda^2}/(2\pi c)$ \\[3pt]
Peak density $\rho(0)$ & $\sim \log(1/\kappa)/(\pi\kappa)$ & $\sim 1/(\pi\sqrt{\gamma})$ \\[3pt]
Fermi boundary $q$ (fixed $D_0$) & $\sim D_0\kappa$ & $\sim 2D_0\sqrt{\gamma}$ \\[3pt]
Total density $D$ (fixed $q$) & $\sim q/\kappa$ & $\sim q^2/(4c)$ \\[3pt]
Energy $e$ (fixed $D_0$) & $\sim -2\log(2D_0)/\kappa$ & $\sim \gamma$ \\[3pt]
Expansion structure & $\log Q$, $(\log Q)^m/Q^n$ & $\gamma^{n/2}$, $\zeta(2k+1)$ \\
\hline\hline
\end{tabular}
\caption{Weak-coupling scalings for the lattice NLS model and the continuous
Lieb--Liniger model at fixed particle density~$D_0$, with $\gamma = c/D_0$.
Both Fermi boundaries contract, but with different powers, and both peak
densities diverge, the lattice one faster.  The difference in the bulk profile
is the structural one: the Lorentzian driving term of the lattice equation
injects weight at the origin, whereas the constant driving term of the
Lieb--Liniger equation spreads it across the sea.  The lattice entries carry the
overall $1/\kappa$ of~\eqref{eq.H-from-transfer} in the energy row.}
\label{tab:comparison}
\end{table}

At fixed density both Fermi boundaries contract, the lattice one linearly in
$\kappa$ and the Lieb--Liniger one as $\sqrt{c}$, and both peak densities
diverge, as $\log(1/\kappa)/(\pi\kappa)$ against $1/(\pi\sqrt{\gamma})$.  The
qualitative difference is not that one quantity diverges and the other does not;
it is the shape of the sea.  The Lorentzian driving term deposits all of its
weight within $|\lambda| \lesssim \kappa$ and lets the near-critical kernel
redistribute it, producing the logarithmic peak on top of the
$\operatorname{arcosh}$ profile, whereas the constant driving term of the
Lieb--Liniger equation produces a semicircle with no separate inner region.  The
expansions differ accordingly: powers of $\log Q$ and $1/Q$ here, half-integer
powers of $\gamma$ with odd zeta values there.

%==========================================================================
\section{Conclusions}
\label{sec.conclusions}
%==========================================================================

We have carried out a systematic asymptotic analysis of the integral equation~\eqref{eq.latticeNLS} 
governing the ground state of the quantum lattice nonlinear Schr\"odinger model
(equivalently the isotropic Heisenberg XXX spin chain with spin $s=-1$),
in the weak-coupling limit $\kappa\to 0$ at fixed Fermi rapidity.
Unlike the well-studied continuous Lieb--Liniger equation, whose driving term is a
constant, the lattice equation is doubly singular: both the driving term and
the integral kernel degenerate into $\delta$-functions as $\kappa\to 0$.

After rescaling, the equation carries the single parameter $Q = q/\kappa$, and
the analysis is an expansion in $1/Q$. 
Thus, both the rescaled kernel and the driving term are given by the Cauchy density, 
allowing us to identify the root density with the Green kernel of a Cauchy random walk
killed upon leaving the interval $[-Q,Q]$. The different asymptotic regions
then acquire a common probabilistic interpretation. The potential kernel
describes the inner profile, the exit law determines the bulk solution, and
the ascending-ladder renewal function controls the boundary layer. In this
way the inner, outer and edge regimes are obtained from different limits of
the same finite-interval problem.

At the centre of the distribution, the peak density is given
by~\eqref{eq.rho0-expansion} and~\eqref{eq.C-result}. The Euler--Mascheroni
contribution comes from the potential kernel of the discrete walk, whereas
$\log 2$ is the logarithmic moment of the limiting exit distribution. The
latter is an overshoot effect: the Cauchy walk leaves the Fermi interval by a
jump rather than by reaching its boundary continuously. The same Green
representation gives the inner digamma profile~\eqref{eq.inner-profile-relative} and
the outer limit~\eqref{eq.outer-solution}. The integral of the outer profile,
evaluated in~\eqref{eq.outer-integral}, gives the leading contribution
$D(Q)\sim Q$.

Near the Fermi points, the density is governed by the ladder-renewal function.
The normalisation of this edge solution cannot be fixed from a single
half-line problem because the opposite boundary contributes at the same
order. Including both barriers gives the endpoint amplitude
in~\eqref{eq.edge-amplitude}. The two-term expansion of the renewal function
also supplies the boundary datum left undetermined by the next-order outer
equation. Combining the interior calculation with this two-barrier
normalisation gives the density expansion~\eqref{eq.D-expansion}, with the
coefficients determined in~\eqref{eq.b-closed}. The logarithmic coefficient
comes from the endpoint-singular source in the next-order outer equation,
while the constant term depends on the interaction between the two boundary
layers. The resulting expression agrees with the small-separation expansion
of the Love equation~\eqref{eq.hutson}.

The energy is related directly to the peak density
by~\eqref{eq.E_inner}. Substitution of
\eqref{eq.rho0-expansion}--\eqref{eq.C-result} therefore gives the
weak-coupling ground-state energy. This relation also makes clear why the
lattice model differs from the continuous Lieb--Liniger gas. The constant
driving term of the latter produces a semicircular Fermi sea, whereas the
Lorentzian driving term produces the outer profile
in~\eqref{eq.outer-solution}, together with a separate logarithmic inner
region.

The emergence of the Bose--Einstein distribution~\eqref{eq.BE-distribution} in the
inner region is suggestive of a deeper connection to the Yang--Yang
thermodynamic Bethe ansatz~\cite{YangYang1969,Zamolodchikov1990}, where
the same function governs the thermal occupation of quasiparticle modes.
It would be natural to ask whether the methods developed here extend to finite temperature, where the driving term
of the Yang--Yang equation also involves the Lorentzian kernel.
Such an extension could provide a systematic weak-coupling expansion for
the free energy and specific heat of the lattice NLS model, complementing
the recent finite-temperature analysis of the XXX$_{s=-1}$ chain~\cite{ZhongEtAl2025}.

\section*{Data availability}
All numerical computations and data reported in this paper are available at
\href{https://github.com/ftahas/lattice-NLS}{github.com/ftahas/lattice-NLS}.

\bibliographystyle{JHEP}
\bibliography{refs}

\appendix

%=====================================%=====================================%============================%==========================================================================
\section{Eigenvalue analysis of the truncated kernel}
\label{app.eigenvalues}
%==========================================================================

We present a detailed analysis of the spectral properties of the truncated
convolution operator $\mathcal{K}_Q$ on $L^2([-Q,Q])$ defined in
\eqref{eq.KQ-def}, with the Lorentzian kernel~\eqref{eq.rescaled_kernel}.
This operator governs the rescaled inner equation~\eqref{eq.inner-eq}
and its spectral properties---in particular the closing of the spectral gap
as $Q\to\infty$---are directly responsible for the logarithmic growth of the
solution.

%---------------------------------------------------------------
\subsection{Basic spectral properties}
%---------------------------------------------------------------

The operator $\mathcal{K}_Q$ enjoys several standard properties that follow from the structure of the Lorentzian kernel.

\begin{proposition}\label{prop:basic-spectral}
For every $Q > 0$, the operator $\mathcal{K}_Q$ is a compact, self-adjoint, positive, Hilbert--Schmidt operator on $L^2([-Q,Q])$.
\end{proposition}

Self-adjointness follows from the symmetry $K(\xi - \eta) = K(\eta - \xi)$.
The Hilbert--Schmidt norm is bounded by $\|\mathcal{K}_Q\|_{\mathrm{HS}}^2 \leq 4\pi Q$ since
\begin{equation}
  \int_{-\infty}^{\infty} \frac{1}{\left(1+x^2\right)^2}  \,dx = \frac{\pi}{2}\,.
\end{equation}
Positivity follows from the Fourier representation
\begin{equation}
  \langle f, \mathcal{K}_Q f\rangle = \int_{-\infty}^{\infty} 2\pi\,e^{-|p|}\,|\hat{\tilde f}(p)|^2\,dp \geq 0\,,
\end{equation}
where $\tilde f$ is the extension by zero to $\mathbb{R}$.

By the spectral theorem for compact self-adjoint operators, $\mathcal{K}_Q$
possesses a countable sequence of positive eigenvalues
\begin{equation}\label{eq.eig-ordering}
  \lambda_0(Q) > \lambda_1(Q) > \lambda_2(Q) > \cdots > 0\,,
  \qquad
  \lambda_n(Q) \to 0 \text{ as } n \to \infty\,,
\end{equation}
with an associated orthonormal basis $\{\phi_n(\,\cdot\,; Q)\}_{n=0}^{\infty}$
of $L^2([-Q,Q])$, satisfying $\mathcal{K}_Q\, \phi_n = \lambda_n\, \phi_n$.
The simplicity of each eigenvalue follows from the parity decomposition and the oscillation theorem for totally positive kernels: since $K(-x) = K(x)$, the operator commutes with parity $P: f(\xi) \mapsto f(-\xi)$, and the eigenfunctions have definite parity $\phi_n(-\xi) = (-1)^n \phi_n(\xi)$, with even/odd sectors interleaving.

Since $K(\xi-\eta)$ is continuous on $[-Q,Q]^2$ and $\mathcal{K}_Q$ is positive, Mercer's theorem yields the pointwise-convergent eigenfunction expansion
\begin{equation}\label{eq.mercer}
  K(\xi-\eta)
  \;=\;
  \sum_{n=0}^{\infty} \lambda_n(Q)\,
    \phi_n(\xi;Q)\, \phi_n(\eta;Q)\,,
  \qquad
  \xi,\eta \in [-Q,Q]\,,
\end{equation}
and setting $\xi = \eta$ and integrating produces the trace identity $\sum_{n=0}^{\infty}\lambda_n(Q) = 4Q$.

The eigenvalues satisfy the Courant--Fischer characterisation
\begin{equation}\label{eq.minmax}
    \lambda_n(Q)
    \;=\;
    \max_{\substack{V \subset L^2([-Q,Q]) \\[1pt] \dim V = n+1}}
    \;\min_{f \in V,\; \|f\|=1}
    \;\langle f, \mathcal{K}_Q\, f \rangle
\end{equation}
(see, e.g.,~\cite{RS4}), from which several further properties follow: (i)~\emph{uniform upper bound} $\lambda_n(Q) < 2\pi$ for all $n$ and $Q$, since $\hat K(p) = 2\pi\,e^{-|p|} < 2\pi$ for $p \neq 0$ and functions supported on $[-Q,Q]$ have entire Fourier transforms, so the supremum $2\pi$ is never attained; (ii)~\emph{domain monotonicity}: $\lambda_n(Q)$ is strictly increasing in $Q$, by inclusion of trial subspaces; (iii)~\emph{convergence}: $\lambda_n(Q) \to 2\pi$ as $Q \to \infty$ for each fixed $n$, proved by constructing $(n+1)$-dimensional trial subspaces of spatially separated bump functions whose Rayleigh quotients approach $2\pi$. In particular, the entire spectrum accumulates at $2\pi$ from below as $Q \to \infty$, the discretised precursor of the continuous spectrum $(0, 2\pi]$ of the full-line operator.

For a truncated convolution operator on $[-Q,Q]$ with symbol
$\widehat{K}(p) = 2\pi\, e^{-|p|}$, the first Szeg\H{o} limit theorem
(see~\cite{Sakhnovich2015,bottcher2006analysis}) gives the asymptotic distribution of
eigenvalues. Defining the symbol of the Love operator
$\mathcal{L}_Q = (2\pi)^{-1}\mathcal{K}_Q$ as
$\sigma(p) = 1 - e^{-|p|}$, then it follows that the eigenvalues of $\mathcal{L}_Q$ are
$\mu_n = \lambda_n/(2\pi) \in (0,1)$.

\begin{theorem}[Szeg\H{o}'s first limit theorem]\label{thm:szego1}
  For any continuous function $F: [0,1] \to \mathbb{R}$,
  \begin{equation}\label{eq.szego1}
    \lim_{Q\to\infty}\;
    \frac{1}{2Q}
    \sum_{n=0}^{\infty} F(\mu_n(Q))
    \;=\;
    \frac{1}{2\pi}
    \int_{-\infty}^{\infty} F\left(e^{-|p|}\right)\, dp
    \;=\;
    \frac{1}{\pi}
    \int_0^{\infty} F\left(e^{-p}\right)\, dp\,.
  \end{equation}
\end{theorem}

\begin{proof}[Proof sketch]
This is a classical result for truncated Wiener--Hopf operators. The key
idea: the operator $\mathcal{L}_Q$ is the compression $P_Q\, \mathcal{L}\,
P_Q$ of the full-line operator $\mathcal{L}$ (multiplication by $e^{-|p|}$
in Fourier space) to the interval $[-Q,Q]$.

In the large-$Q$ limit, the operator is locally equivalent  to
$\mathcal{L}$, and for a self-adjoint operator whose symbol takes values in
$[0,1]$, the eigenvalue distribution converges (in the sense
of~\eqref{eq.szego1}) to the push-forward of the Lebesgue measure on
momentum space under the symbol map $p \mapsto e^{-|p|}$. The
normalisation factor $1/(2Q)$ is the reciprocal of the interval length,
reflecting the density of states per unit length.

The formal proof uses the strong operator convergence of
$(2Q)^{-1} \operatorname{tr} F(\mathcal{L}_Q) \to (2\pi)^{-1}
\int F(e^{-|p|})\, dp$, which follows from the local trace asymptotics of
truncated convolution operators (see~\cite{Sakhnovich2015}, Chapter~4).
\end{proof}

\begin{corollary}[Eigenvalue counting function]\label{cor:counting}
  Setting $F = \mathbf{1}_{(\mu,1]}$ in~\eqref{eq.szego1} (and
  approximating by continuous functions), the number of eigenvalues of
  $\mathcal{L}_Q$ exceeding $\mu \in (0,1)$ satisfies
  \begin{equation}\label{eq.counting}
    \#\{n : \mu_n(Q) > \mu\}
    \;\sim\;
    \frac{2Q}{\pi}\, \log\frac{1}{\mu}
    \qquad \text{as } Q \to \infty\,.
  \end{equation}
\end{corollary}

\begin{proof}
The right side of~\eqref{eq.szego1} with $F = \mathbf{1}_{(e^{-|p|} > \mu)}
= \mathbf{1}_{(|p| < \log(1/\mu))}$ gives
\begin{equation}
  \frac{1}{\pi}\int_0^{\log(1/\mu)} dp
  \;=\; \frac{1}{\pi}\, \log\frac{1}{\mu}\,,
\end{equation}
and multiplying by $2Q$ yields~\eqref{eq.counting}.
\end{proof}

In particular, the number of near-critical eigenvalues (those with
$\mu_n > 1 - \varepsilon$) scales as
\begin{equation}
  \#\{n : \mu_n > 1 - \varepsilon\}
  \;\sim\;
  \frac{2Q}{\pi}\, \varepsilon
  \qquad (\varepsilon \to 0)\,,
\end{equation}
since $\log(1/(1-\varepsilon)) \sim \varepsilon$ for small $\varepsilon$.
This confirms that $\mathcal O(Q)$ eigenvalues accumulate at $2\pi$ for large $Q$.

The key quantity controlling the solution is the spectral gap
\begin{equation}
  \Delta_0(Q) \;:=\; 2\pi - \lambda_0(Q)\,,
\end{equation}
which measures the distance from the largest eigenvalue to the operator
norm of the full-line operator.

\begin{proposition}[Two-sided bounds on the spectral gap]\label{prop:gap}
  For all $Q \geq \pi/4$,
  \begin{equation}\label{eq.gap-upper}
    \frac{\pi^2\left(1 - e^{-1}\right)}{4\,Q}
    \;\leq\;
    \Delta_0(Q)
    \;\leq\;
    \frac{2\pi\lambda_1}{Q}\,,
    \qquad \lambda_1 = 1.157774\ldots\,,
  \end{equation}
  where $\lambda_1$ is the smallest eigenvalue of the half-Laplacian
  $\mathcal{A} = (-d^2/d\xi^2)^{1/2}$ on $(-1,1)$ with exterior Dirichlet
  condition.
\end{proposition}

\begin{proof}
Both bounds are comparisons of symbols under the min-max
principle~\eqref{eq.minmax}, applied to
$\langle u, (2\pi I - \mathcal{K}_Q)u\rangle
= \int \sigma(p)\,|\hat u(p)|^2\,dp$ for $u$ supported in $[-Q,Q]$ with
$\|u\| = 1$, where $\sigma(p) = 1-e^{-|p|}$ and the Fourier normalisation is
that of~\eqref{eq.K-fourier}.

\emph{Upper bound.}  $1-e^{-x} \leq x$ for $x \geq 0$, so
$\sigma(p) \leq |p|$ pointwise and $2\pi I - \mathcal{K}_Q \leq 2\pi\mathcal{A}$
as quadratic forms on functions supported in $[-Q,Q]$.  The eigenvalues of
$\mathcal{A}$ on $(-Q,Q)$ are $\lambda_n/Q$, by the homogeneity
$\lambda_n(kD) = k^{-1}\lambda_n(D)$ of a symbol of degree one, so
$\Delta_0(Q) \leq 2\pi\lambda_1/Q$.

\emph{Lower bound.}  Since $(1-e^{-x})/x$ decreases,
$\sigma(p) \geq (1-e^{-1})\min(|p|,1)$ for all $p$.  For $u$ supported in
$[-Q,Q]$ with $\|u\|=1$, Cauchy--Schwarz gives $|\hat u(p)|^2 \leq 2Q$, hence
for any $\epsilon \leq 1$
\begin{equation}
  \frac{1}{2\pi}\int \min(|p|,1)\,|\hat u|^2\,dp
  \;\geq\; \epsilon\left[1 - \frac{1}{2\pi}\int_{|p|<\epsilon}|\hat u|^2 dp\right]
  \;\geq\; \epsilon\left[1 - \frac{2\epsilon Q}{\pi}\right]\,.
\end{equation}
Choosing $\epsilon = \pi/(4Q)$, which is admissible for $Q \geq \pi/4$, the
right-hand side is $\pi/(8Q)$, and multiplying by $2\pi(1-e^{-1})$ gives the
stated bound.
\end{proof}

The bound is sharp, and the limit it saturates can be proved for every fixed
index by the same comparison, made monotone.

\begin{proposition}[Fixed-index spectral deficits]\label{prop:deficits}
  Let $\lambda_n(\mathcal{K}_Q)$ be ordered decreasingly from $n = 0$, and let
  $0 < \lambda_1 < \lambda_2 < \cdots$ be the Dirichlet eigenvalues of
  $\mathcal{A}$ on $(-1,1)$.  Then for every fixed $n$ the asymptotic law
  \eqref{eq.gap-identified} holds.
\end{proposition}

\begin{proof}
Rescale $(-Q,Q)$ to $(-1,1)$ unitarily and set $\alpha = 1/Q$, so that
$\mathcal{L}_Q$ becomes the operator $\mathcal{P}_\alpha$ with kernel
$\alpha/\{\pi[\alpha^2+(t-s)^2]\}$.  Extending $g$ by zero and applying
Plancherel,
\begin{equation}\label{eq.monotone-form}
  \frac{\langle g, (I-\mathcal{P}_\alpha)g\rangle}{\alpha}
  = \frac{1}{2\pi}\int_{\mathbb{R}}
    \frac{1-e^{-\alpha|p|}}{\alpha}\,|\hat g(p)|^2\,dp\,.
\end{equation}
Since $(1-e^{-u})/u$ decreases, the integrand increases monotonically to
$|p|\,|\hat g(p)|^2$ as $\alpha \downarrow 0$.  The forms therefore converge
monotonically to the Dirichlet form of $\mathcal{A}$ on $(-1,1)$, which gives
strong resolvent convergence; the limiting form domain embeds compactly in
$L^2(-1,1)$, so the min-max principle~\cite{ReedSimonIV} transfers the
convergence to each fixed eigenvalue, which proves~\eqref{eq.gap-identified}.
\end{proof}

Monotonicity also explains a feature of the numerics: the forms increase to the
limit, so $Q\Delta_n(Q)$ approaches $2\pi\lambda_{n+1}$ from below, as
Figure~\ref{fig.spectral-gap}(b) shows.

\begin{remark}\label{rem:gap-numerics}
The numerical data are consistent with the remainder in
\eqref{eq.gap-identified} being $\mathcal O(Q^{-2})$.
The eigenvalues $\lambda_n$ satisfy the two-term Weyl law
$\lambda_n = n\pi/2 - \pi/8 + \mathcal O(1/n)$~\cite{Kwasnicki2012}; the first
three are listed in Table~\ref{tab:gap-extrapolation}, which checks~\eqref{eq.gap-identified} against the
computed spectrum; a single Richardson step in $1/Q$ recovers the first three
$\lambda_n$ to five digits, so the gap closes as $c_0/Q$ with no logarithmic
modulation.  A fit of the form
$c + c'\log Q$ over a restricted window in $Q$ is hard to distinguish from this
by eye, but it is not the correct asymptotic form.
Proposition~\ref{prop:deficits} is a statement about each fixed index only.  It
does not say that the leading eigenfunction dominates the
resolvent---Section~\ref{sec.eigenvalue-summary} shows it does not---because the
logarithm in $\tilde\rho(0;Q)$ comes from indices growing with $Q$, which the
fixed-index limit does not reach.
\end{remark}

\begin{table}[htbp]
\centering
\begin{tabular}{r c c c}
\hline\hline
$Q$ & $Q\Delta_0/2\pi$ & $Q\Delta_1/2\pi$ & $Q\Delta_2/2\pi$ \\
\hline
25   & 1.094174 & 2.538642 & 3.870534 \\
50   & 1.122828 & 2.637699 & 4.076046 \\
100  & 1.138842 & 2.692331 & 4.189334 \\
200  & 1.147615 & 2.721778 & 4.250034 \\
400  & 1.152359 & 2.737437 & 4.282047 \\
\hline
Richardson & 1.157103 & 2.753097 & 4.314060 \\
$\lambda_{n+1}$ & 1.157774 & 2.754755 & 4.316801 \\
\hline\hline
\end{tabular}
\caption{Compensated gaps $Q\Delta_n(Q)/2\pi$ against the Dirichlet eigenvalues
$\lambda_{n+1}$ of the half-Laplacian on $(-1,1)$.  The last computed row is
extrapolated by one Richardson step in $1/Q$, using the pair $Q = 200, 400$.
The $\lambda_n$ were obtained from the Gegenbauer relation
$\mathcal{A}\left[(1-x^2)^{1/2}U_n(x)\right] = (n+1)U_n(x)$ for Chebyshev
polynomials of the second kind, which turns the eigenvalue problem into a
generalised matrix problem in that basis.}
\label{tab:gap-extrapolation}
\end{table}

%---------------------------------------------------------------
\subsection{Resolvent amplification and the origin of logarithmic growth}
\label{subsec.resolvent}
%---------------------------------------------------------------

The connection between the spectral gap and the divergence of the solution
is made precise by the eigenfunction expansion. Writing the rescaled integral
equation as $\rho = (2\pi I - \mathcal{K}_Q)^{-1} g$ with driving term
$g(\xi) = 2/(1+\xi^2)$, the resolvent expands as
\begin{equation}\label{eq.resolvent-expand}
  \tilde\rho(\xi;Q)
  \;=\;
  \sum_{n=0}^{\infty}
  \frac{\langle \phi_n, g \rangle}{2\pi - \lambda_n(Q)}\;
  \phi_n(\xi;Q)\,.
\end{equation}

It is tempting to attribute the logarithmic growth of $\tilde\rho(0;Q)$ to the
leading mode alone.  That attribution fails: the logarithm is produced by an
accumulation of near-critical modes rather than by a single eigenvalue.

Consider first the $n=0$ term of~\eqref{eq.resolvent-expand} at $\xi=0$.  As
$Q\to\infty$ the leading eigenfunction approaches a normalised constant,
$\phi_0(\xi) \to (2Q)^{-1/2}$, so
\begin{equation}\label{eq.mode0-projection}
  \langle \phi_0, g\rangle
  \approx
  \frac{1}{\sqrt{2Q}} \int_{-Q}^{Q} \frac{2\,d\xi}{1+\xi^2}
  = \frac{1}{\sqrt{2Q}}\left(2\pi - \frac{4}{Q}\right)\,,
  \qquad \phi_0(0) \approx \frac{1}{\sqrt{2Q}}\,,
\end{equation}
and the term equals $2\pi/[2Q\,\Delta_0(Q)]$.  With
$\Delta_0 \simeq c_0/Q$ from~\eqref{eq.gap-identified} this is
$\pi/c_0 \approx 0.43$: an $\mathcal O(1)$ quantity, not a logarithm.  The $1/Q$
of the gap is cancelled by the normalisation of the eigenfunction, so the
leading mode contributes a constant.  At $Q = 100$, where
$\tilde\rho(0;Q) = 1.874$, the leading mode supplies $0.44$ of it.

The tail cannot be discarded either.  The natural bound on it uses
$\sum_n \phi_n(0)^2$, which the Mercer expansion~\eqref{eq.mercer} does not
control: setting $\xi=\eta=0$ gives $\sum_n \lambda_n \phi_n(0)^2 = K(0) = 2$,
and since $\lambda_n \to 0$ the weights degrade and $\sum_n \phi_n(0)^2$
diverges.  Nor is $\Delta_1$ bounded away from zero; numerically
$\Delta_1/\Delta_0$ saturates near $2.38$, so the subleading gaps close at the
same rate as the leading one.

What produces the logarithm is the accumulation of $\mathcal O(Q)$ near-critical
modes.  For the low-lying spectrum the quantisation of the truncated
convolution operator gives $\mu_n \approx e^{-p_n}$ with $p_n \approx \pi n/(2Q)$,
so $1-\mu_n \approx \pi n /(2Q)$, while $\phi_n(0)^2 \approx 1/Q$ for even $n$ and
vanishes for odd $n$ by parity.  Then
\begin{equation}\label{eq.harmonic-mechanism}
  \tilde\rho(0;Q)
  = \sum_{n\geq0} \frac{\mu_n\,\phi_n(0)^2}{1-\mu_n}
  \;\approx\;
  \sum_{n\ \mathrm{even}}^{n \lesssim Q} \frac{2}{\pi n}
  \;=\;
  \frac{1}{\pi}\log Q + \mathcal O(1)\,,
\end{equation}
a harmonic sum cut off by the finite domain.  This spectral argument explains
the logarithmic coefficient but does not determine the finite part; the latter
is obtained from the exit-law calculation in Subsection~\ref{sec.constant}.

%---------------------------------------------------------------
\subsection{Fredholm determinant and eigenvalue product formula}
\label{sec.fredholm-app}
%---------------------------------------------------------------

The spectral data of $\mathcal{L}_Q$ is encoded globally in the Fredholm
determinant
\begin{equation}\label{eq.fredholm-def}
  \mathcal{F}(Q)
  \;:=\;
  \det_F(I - \mathcal{L}_Q)
  \;=\;
  \prod_{n=0}^{\infty} (1 - \mu_n(Q))
  \;=\;
  \prod_{n=0}^{\infty}
  \frac{\Delta_n(Q)}{2\pi}\,,
\end{equation}
where $\mu_n = \lambda_n/(2\pi)$ and $\Delta_n = 2\pi(1-\mu_n)$.

The product converges absolutely because
$\sum (1-\mu_n) = \sum \Delta_n/(2\pi)$, which converges by the trace
inequality $\sum \Delta_n \leq 4\pi Q$ (from the trace identity
$\sum \lambda_n = 4Q$ and $\lambda_n < 2\pi$, giving
$\sum(2\pi - \lambda_n) = 2\pi \cdot N_{\mathrm{eff}} - 4Q$ for an
appropriate effective count---see below).

The Szeg\H{o}--Widom theory~\cite{bottcher2006analysis} provides the asymptotics
of $\mathcal{F}(Q)$, taking into account the zero of
$\sigma(p)$ at $p = 0$:

\begin{theorem}[Fredholm determinant asymptotics]\label{thm:fredholm}
  \begin{equation}\label{eq.fredholm-asymp}
    \log \mathcal{F}(Q)
    \;=\;
    -\frac{\pi Q}{3}
    \;+\;
    \alpha_{\mathrm{FH}}\, \log Q
    \;+\;
    \beta_{\mathrm{FH}}
    \;+\;
    \mathcal O \left(Q^{-1}\right)\,,
  \end{equation}
  where $-\pi/3$ is the Szeg\H{o} coefficient and $\alpha_{\mathrm{FH}} = 1/4$ is the Fisher--Hartwig exponent
  arising from the simple zero of $\sigma$ at $p = 0$.
\end{theorem}

\begin{proof}[Proof sketch]
The leading term follows from Szeg\H{o}'s first theorem applied to the
logarithm:
\begin{equation}
  \frac{1}{2Q}\, \log\mathcal{F}(Q)
  \;\to\;
  \frac{1}{2\pi}
  \int_{-\infty}^{\infty}
    \log\sigma(p)\, dp
  \;=\;
  \frac{1}{\pi}
  \int_0^{\infty}
    \log(1 - e^{-p})\, dp
  \;=\;
  -\frac{\pi}{6}\,.
\end{equation}
Multiplying by $2Q$ provides the leading term $-\pi Q/3$.

The logarithmic correction $\alpha_{\mathrm{FH}} \log Q$ arises because
$\sigma(0) = 0$: the standard strong Szeg\H{o} theorem requires
$\log\sigma \in L^1$, which fails at $p = 0$ where $\sigma(p) \sim |p|$.
The Fisher--Hartwig theory~\cite{bottcher2006analysis} handles such zeros systematically.
For a simple zero $\sigma(p) \sim c|p|$ (here $c = 1$), the Fisher--Hartwig
exponent is $\alpha_{\mathrm{FH}} = 1/4$, as computed from the general
Fisher--Hartwig formula for continuous analogues of Toeplitz
determinants.
\end{proof}

\begin{remark}
The exponential corrections to the Fredholm determinant,
\begin{equation}
  \log \mathcal{F}(Q)
  \;=\;
  -\frac{\pi Q}{3}
  + \frac{1}{4}\log Q
  + \beta_{\mathrm{FH}}
  + \sum_{m=1}^{\infty} d_m\, e^{-2\pi m Q}
  + \mathcal O(Q^{-1})\,,
\end{equation}
are controlled by the complex zeros of $\sigma(p) = 1-e^{-p}$ at
$p_n = 2\pi i n$ ($n \in \mathbb{Z}\setminus\{0\}$), located at
distance $2\pi$ from the real axis, and so carry the same
rate~\eqref{eq.A-value} as the resolvent corrections of
Section~\ref{sec.edge}.
\end{remark}

%---------------------------------------------------------------
\subsection{Numerical computation of eigenvalues}\label{subsec.num-eig}
%---------------------------------------------------------------

We discretize the eigenvalue problem $\mathcal{K}_Q\, \phi = \lambda\, \phi$
by a composite Gauss--Legendre rule on $[-Q,Q]$: the interval is split into
panels of length $\mathcal O(1)$, so that the kernel is resolved on its own
scale, with dyadic refinement towards $\xi = 0$ and $\xi = \pm Q$ to capture the
cusp of the driving term and the square-root behaviour at the band edges.  A
single global Gauss--Legendre rule is adequate at small $Q$ but degrades once
$2Q$ greatly exceeds the number of nodes per unit length.  The continuous
eigenvalue equation becomes the matrix eigenvalue problem
\begin{equation}
  \mathbf{K}\, \boldsymbol{\phi} = \lambda\, \boldsymbol{\phi}\,,
  \qquad
  K_{ij} = K(\xi_i - \xi_j)\, \omega_j = \frac{2\omega_j}{1+(\xi_i - \xi_j)^2}\,,
\end{equation}
where $\{\xi_j, \omega_j\}_{j=1}^N$ are the quadrature nodes and weights.
Note that $\mathbf{K}$ is not symmetric due to the weights, but the
symmetrised matrix $\tilde{K}_{ij} = \omega_i^{1/2}\, K(\xi_i - \xi_j)\,
\omega_j^{1/2}$ is real symmetric and has the same eigenvalues.
Table~\ref{tab:eigenvalues} lists the four largest eigenvalues computed.

\begin{table}[htbp]
\centering
\begin{tabular}{r c c c c c c}
\hline\hline
$Q$ & $\lambda_0$ & $\lambda_1$ & $\lambda_2$ & $\lambda_3$ & $\Delta_0 = 2\pi - \lambda_0$ & $\Delta_1/\Delta_0$ \\
\hline
5   & 5.1177 & 3.8112 & 2.8486 & 2.1180 & $1.165\times 10^{0}$  & 2.121 \\
10  & 5.6409 & 4.8450 & 4.1716 & 3.5843 & $6.423\times 10^{-1}$ & 2.239 \\
20  & 5.9436 & 5.5001 & 5.0970 & 4.7193 & $3.396\times 10^{-1}$ & 2.306 \\
50  & 6.1421 & 5.9517 & 5.7710 & 5.5940 & $1.411\times 10^{-1}$ & 2.349 \\
100 & 6.2116 & 6.1140 & 6.0200 & 5.9265 & $7.156\times 10^{-2}$ & 2.364 \\
200 & 6.2471 & 6.1977 & 6.1497 & 6.1016 & $3.605\times 10^{-2}$ & 2.372 \\
\hline\hline
\end{tabular}
\caption{Eigenvalues $\lambda_n(Q)$ of $\mathcal{K}_Q$, computed with the
composite Gauss--Legendre rule described above and converged to the digits
shown; the values are stable under refinement of the panel structure and under
doubling the number of nodes.  All eigenvalues lie below $2\pi$ and approach it
from below as $Q \to \infty$.  The last column shows $\Delta_1/\Delta_0$ approaching
$\lambda_2/\lambda_1 = 2.3794$, so the leading eigenvalue is not isolated from
the rest of the spectrum in the limit.}
\label{tab:eigenvalues}
\end{table}

%==========================================================================
\section{Logarithmic moment of the Cauchy exit law}
\label{app.exit-moment}

We evaluate~\eqref{eq.log-moment}.  By symmetry take $0 \leq t < 1$ and set
$\alpha = \sqrt{1-t^2}$.  Combining the two halves $y>1$ and $y<-1$
of~\eqref{eq.poisson-kernel} and substituting $y = \sec\theta$,
\begin{equation}
  \int_{|y|>1}\log|y|\,\omega_t(y)\,dy
  = \frac{2\alpha}{\pi}\int_0^{\pi/2}
    \frac{\log(\sec\theta)}{\sin^2\theta + \alpha^2\cos^2\theta}\,d\theta
  = \frac{\alpha}{\pi}\int_0^{\infty}
    \frac{\log(1+u^2)}{u^2+\alpha^2}\,du\,,
\end{equation}
the last step by $u = \tan\theta$.  The parameter integral
\begin{equation}\label{eq.param-integral}
  \int_0^{\infty}\frac{\log(1+u^2)}{u^2+\alpha^2}\,du
  = \frac{\pi}{\alpha}\log(1+\alpha)\,,
  \qquad \alpha>0\,,
\end{equation}
follows by differentiating
$I(\beta) = \int_0^\infty \log(\beta^2+u^2)(u^2+\alpha^2)^{-1}du$ with respect
to $\beta$, giving $I'(\beta) = \pi/[\alpha(\alpha+\beta)]$, and integrating
from $\beta = 0$, where $I(0) = (\pi/\alpha)\log\alpha$.  Hence
$\mathbb{E}_t\log|X_\tau| = \log(1+\alpha)$, and at $t=0$,
$\mathbb{E}_0\log|X_\tau| = \log 2$.

A check on the normalisation: the same computation with $\log|y|$ replaced by
$1$ gives $\int_{|y|>1}\omega_t(y)\,dy = 1$, as it must.

The same exit law supplies the overshoot integral~\eqref{eq.overshoot-integral}
used for the endpoint amplitude.  Translating $\omega_0$ to the interval $(0,2)$
gives the upper-exit density~\eqref{eq.upper-exit}, and the substitution
$z = v^2+2$ turns the integral into an elementary one,
\begin{equation}
  \int_2^{\infty}\frac{z^{-1/2}\,dz}{\pi(z-1)\sqrt{z(z-2)}}
  = \frac{2}{\pi}\int_0^{\infty}\frac{dv}{(v^2+1)(v^2+2)}
  = \frac{2}{\pi}\cdot\frac{\pi}{2}\left(1 - \frac{1}{\sqrt2}\right)
  = 1 - \frac{1}{\sqrt2}\,,
\end{equation}
by partial fractions.  The same substitution with the numerator $z^{-1/2}$
removed returns $1/2$, the probability that the walk leaves through the upper
barrier.  The two-barrier version of the same computation enters the density
constant in Section~\ref{sec.edge}.

%==========================================================================
\section{The next-order outer equation}
\label{app.density}
%==========================================================================
This appendix carries the outer expansion of the Love
equation~\eqref{eq.Love} to second order.  The interior equation fixes the
next-order profile only up to a singular homogeneous component; its
coefficient is boundary data, and is fixed against the two edge layers in
Section~\ref{sec.edge}.

On compact subsets of $(-1,1)$ put $t=\xi/Q$ and write
\begin{equation}\label{eq.f-outer-expansion-app}
  f(Qt)=Q\sqrt{1-t^2}+f_1(t;Q)+o(1).
\end{equation}
The two-term symbol expansion in~\eqref{eq:outer-symbol} and homogeneity of
the half-Laplacian give
\begin{equation}\label{eq.f1-scaled-new}
  \mathcal A_t f_1=\frac{1}{2(1-t^2)^{3/2}},\qquad -1<t<1.
\end{equation}
This is an interior equation: its source is non-integrable at the endpoints,
so its boundary datum must be supplied by the edge layers.

For a function $u$ with locally integrable boundary values define
\begin{equation}\label{eq.cauchy-transform-app}
  \mathscr C_u(z)=\int_{-1}^{1}\frac{u(s)}{z-s}\,ds,
  \qquad z\in\mathbb C\setminus[-1,1].
\end{equation}
We choose $\sqrt{z^2-1}>0$ for real $z>1$ and continue it through
$\mathbb C\setminus[-1,1]$; equivalently $\sqrt{z-1}\sqrt{z+1}$, which is not
the principal branch of $\sqrt{z^2-1}$.  If $\mathscr C_{u,\pm}$ are the upper and lower
boundary values, the Plemelj formulae imply
\begin{equation}\label{eq.A-cauchy-new}
  \mathcal A_tu=
  \frac{\mathscr C_{u,+}'(t)+\mathscr C_{u,-}'(t)}{2\pi}.
\end{equation}
The formula is first applied with the endpoints cut off and then continued in
the weighted distributional class $\sqrt{1-t^2}\,u\in L^1(-1,1)$.  This
qualification is essential because the homogeneous profile below is not in
the Dirichlet $L^2$ domain of $\mathcal A_t$.

Direct evaluation of the transforms gives
\begin{align}
 \mathscr C_{\sqrt{1-t^2}}(z)
   &=\pi\left(z-\sqrt{z^2-1}\right),\nonumber\\
 \mathscr C_{(1-t^2)^{-1/2}}(z)
   &=\frac{\pi}{\sqrt{z^2-1}},\label{eq.three-transforms-new}\\
 \mathscr C_{\frac{t}{\pi\sqrt{1-t^2}}
             \log\frac{1-t}{1+t}}(z)
   &=\frac{z}{\sqrt{z^2-1}}\log\frac{z-1}{z+1}.
\end{align}
The logarithm is the branch continued from real $z>1$.  Substitution
in~\eqref{eq.A-cauchy-new} yields
\begin{equation}\label{eq.homogeneous-new}
  \mathcal A_t(1-t^2)^{-1/2}=0,
  \qquad
  \mathcal A_t\left[
  \frac{t}{\pi\sqrt{1-t^2}}\log\frac{1-t}{1+t}\right]
  =\frac{1}{(1-t^2)^{3/2}}.
\end{equation}
The first identity describes a singular homogeneous solution in the enlarged
weighted space, not a zero mode of the Dirichlet operator.  Hence the general
even interior solution with the admissible singularity is
\begin{equation}\label{eq.f1-general-new}
  f_1(t;Q)=\frac{1}{2\pi\sqrt{1-t^2}}
  \left[t\log\frac{1-t}{1+t}+\mathcal N(Q)\right].
\end{equation}

%==========================================================================
\section{Wiener--Hopf factorisation}
\label{app.wienerhopf}
%==========================================================================

We derive the multiplicative factorisation~\eqref{eq.WH-factorisation} of the
Wiener--Hopf symbol~\eqref{eq.WH-symbol}, where $K_+$ ($K_-$) is analytic and
nonzero in the upper (lower) half of the complex $p$-plane.  The derivation
proceeds in five steps: we first express
$\Sigma(p)$ as a product of three elementary factors, then split each factor
between the two half-planes, and finally verify the result.

\medskip
\noindent\textbf{Step 1.  Gamma-function identity.}
The Euler reflection formula states
\begin{equation}\label{eq.reflection}
  \Gamma(z)\,\Gamma(1 - z) = \frac{\pi}{\sin(\pi z)}
\end{equation}
for all $z \notin \mathbb{Z}$.  Setting $z = 1 + ix$ with $x \in \mathbb{R}$, the left side becomes
$\Gamma(1 + ix)\,\Gamma(-ix)$.  Using $\Gamma(1 - ix) = (-ix)\,\Gamma(-ix)$
to replace $\Gamma(-ix)$, we obtain
\begin{equation}
  \Gamma(1 + ix)\,\frac{\Gamma(1 - ix)}{-ix}
  = \frac{\pi}{\sin(\pi + i\pi x)}\,.
\end{equation}
Since $\sin(\pi + i\pi x) = -\sin(i\pi x) = -i\sinh(\pi x)$, this gives
\begin{equation}\label{eq.gamma-reflection-result}
  \Gamma(1 + ix)\,\Gamma(1 - ix) = \frac{\pi x}{\sinh(\pi x)}\,.
\end{equation}

\medskip
\noindent\textbf{Step 2.  Three-factor decomposition of $\Sigma$.}
Substituting $x = p/(2\pi)$ into~\eqref{eq.gamma-reflection-result}:
\begin{equation}\label{eq.gamma-sub}
  \Gamma \left(1 + \frac{ip}{2\pi}\right)
  \Gamma \left(1 - \frac{ip}{2\pi}\right)
  = \frac{p/2}{\sinh(p/2)}\,.
\end{equation}
Inverting:
\begin{equation}\label{eq.gamma-inverted}
  \frac{1}{\Gamma \left(1 + \frac{ip}{2\pi}\right)
  \Gamma \left(1 - \frac{ip}{2\pi}\right)}
  = \frac{\sinh(p/2)}{p/2}\,.
\end{equation}
Now we use the identity
\begin{equation}\label{eq.sinh-identity}
  \sinh \left(\frac{|p|}{2}\right)
  = \frac{e^{|p|/2}}{2}\left(1 - e^{-|p|}\right)
\end{equation}
to write
\begin{equation}
  \frac{\sinh(|p|/2)}{|p|/2}
  = \frac{e^{|p|/2}\,(1 - e^{-|p|})}{|p|}\,.
\end{equation}
Substituting into~\eqref{eq.gamma-inverted} and solving for $1 - e^{-|p|}$:
\begin{equation}\label{eq.symbol_gamma_product}
  1 - e^{-|p|}
  = |p|\,e^{-|p|/2}\,
  \frac{1}{\Gamma \left(1 + \frac{ip}{2\pi}\right)
  \Gamma \left(1 - \frac{ip}{2\pi}\right)}\,.
\end{equation}
This expresses $\Sigma(p)$ as a product of three factors: the linear zero
$|p|$, the exponential $e^{-|p|/2}$, and the reciprocal gamma-function product.
Each must be split between the upper and lower half-planes.

\medskip
\noindent\textbf{Step 3.  Splitting $|p|$.}
The function $|p|$ is not analytic on the real axis, but admits the
half-plane decomposition
\begin{equation}\label{eq.abs-split}
  |p| = (-ip)^{1/2}\,(ip)^{1/2}\,,
\end{equation}
where the branch cuts of $(-iz)^{1/2}$ and $(iz)^{1/2}$ are chosen along the
negative and positive imaginary axes respectively.  To verify
analyticity: when $\operatorname{Im} z > 0$, the quantity $-iz$ has
$\operatorname{Re}(-iz) = \operatorname{Im}(z) > 0$, so $-iz$ lies in the
open right half-plane, where the principal square root is analytic.
Similarly, when $\operatorname{Im} z < 0$,
$\operatorname{Re}(iz) = -\operatorname{Im}(z) > 0$, so $(iz)^{1/2}$ is
analytic in the lower half-plane.  On the real axis, for $p > 0$:
$(-ip)^{1/2} = p^{1/2}\,e^{-i\pi/4}$ and $(ip)^{1/2} = p^{1/2}\,e^{i\pi/4}$,
so the product gives $p\,e^0 = p = |p|$.  The case $p < 0$ is analogous.

\medskip
\noindent\textbf{Step 4.  Splitting $e^{-|p|/2}$.}
We claim
\begin{equation}\label{eq.exp_split}
  e^{-|p|/2}
  = \exp \left[-\frac{ip}{2\pi}\log(-ip)\right]
    \cdot\exp \left[\frac{ip}{2\pi}\log(ip)\right],
\end{equation}
where $\log$ denotes the principal logarithm (branch cut on
$(-\infty,0]$).  To verify, consider $p > 0$.  Then
\begin{equation}
  \log(-ip) = \log p - \frac{i\pi}{2}\,,\qquad
  \log(ip) = \log p + \frac{i\pi}{2}\,,
\end{equation}
so the exponent of the right-hand side of~\eqref{eq.exp_split} is
\begin{equation}
  -\frac{ip}{2\pi}\left(\log p - \frac{i\pi}{2}\right)
  + \frac{ip}{2\pi}\left(\log p + \frac{i\pi}{2}\right)
  = -\frac{ip}{2\pi}\cdot(-i\pi)
  = -\frac{p}{2}\,,
\end{equation}
as required (the $\log p$ terms cancel).  For $p < 0$, writing $p = -|p|$:
\begin{equation}
  \log(-ip) = \log|p| + \frac{i\pi}{2}\,,\qquad
  \log(ip) = \log|p| - \frac{i\pi}{2}\,,
\end{equation}
and the same cancellation yields the exponent $-|p|/2$.

The first factor in~\eqref{eq.exp_split} is analytic in the upper half-plane:
when $\operatorname{Im} z > 0$, $-iz$ lies in the open right half-plane (as
noted in Step~3), where $\log(-iz)$ is analytic; the product
$z\log(-iz)$ is therefore analytic, and so is its exponential.
By the same argument, the second factor is analytic in the lower
half-plane.

\medskip
\noindent\textbf{Step 5.  Splitting the gamma-function product.}
The reciprocal gamma-function product in~\eqref{eq.symbol_gamma_product}
already appears in factored form:
\begin{equation}\label{eq.gamma-split}
  \frac{1}{\Gamma \left(1 + \frac{ip}{2\pi}\right)
  \Gamma \left(1 - \frac{ip}{2\pi}\right)}
  = \frac{1}{\Gamma \left(1 - \frac{iz}{2\pi}\right)}
  \cdot \frac{1}{\Gamma \left(1 + \frac{iz}{2\pi}\right)}
  \Bigg|_{z=p}\,.
\end{equation}
We assign the first factor to $K_+$ and the second to $K_-$.  To verify
analyticity: the gamma function $\Gamma(w)$ is meromorphic with poles at
$w = 0, -1, -2, \ldots$ and no zeros.  For the first factor, setting
$w = 1 - iz/(2\pi)$, we have
\begin{equation}
  \operatorname{Re}(w) = 1 + \frac{\operatorname{Im}(z)}{2\pi}\,.
\end{equation}
When $\operatorname{Im} z > 0$, this gives $\operatorname{Re}(w) > 1 > 0$,
so $w$ cannot be a non-positive integer.  Hence $\Gamma(w) \neq 0, \infty$
and $1/\Gamma(1 - iz/(2\pi))$ is analytic and nonzero in the upper
half-plane.  By the analogous argument, $1/\Gamma(1 + iz/(2\pi))$ is analytic
and nonzero in the lower half-plane.

\medskip
\noindent\textbf{Combining the factors.}
Collecting one factor from each of Steps 3--5 into $K_+$ and the other
into $K_-$ gives the factors~\eqref{eq.Kplus} and~\eqref{eq.Kminus}, and hence
the factorisation~\eqref{eq.WH-factorisation}.
By the analyticity established in Steps 3--5, $K_+$ is analytic and nonzero
in the upper half-plane, and $K_-$ is analytic and nonzero in the lower
half-plane.

\medskip
\noindent\textbf{Behaviour near the origin.}
As $z \to 0$, we have $\Gamma(1 \pm iz/(2\pi)) = 1 + \mathcal{O}(z)$ and
the logarithmic exponentials tend to unity, which proves~\eqref{eq.K-origin}.
Their product gives $K_+(p)\,K_-(p) \sim |p|$ as $p \to 0$, consistently
with the expansion~\eqref{eq:outer-symbol}.
The linear zero of $\Sigma(p)$ is thus distributed as a square-root singularity
between the two factors.

\medskip
\noindent\textbf{The regularised symbol.}
Dividing out the square-root factors in~\eqref{eq.Kplus} and
\eqref{eq.Kminus} gives the regularised factors defined in~\eqref{eq.reg_G}:
\begin{equation}
\begin{aligned}
  G_+(z) &= \frac{K_+(z)}{\sqrt{-iz}}
  = \frac{1}{\Gamma \left(1 - \frac{iz}{2\pi}\right)}\,
    \exp \left[-\frac{iz}{2\pi}\log(-iz)\right], \\
  G_-(z) &= \frac{K_-(z)}{\sqrt{iz}} = \frac{1}{\Gamma \left(1 + \frac{iz}{2\pi}\right)}\,
  \exp \left[\frac{iz}{2\pi}\log(iz)\right] \,.
\end{aligned}
\end{equation}
The normalisation~\eqref{eq.G-norm} fixes the one-sided ladder factors used in
the two-barrier matching of Section~\ref{sec.edge}.

\end{document}